%% file: alpha2pt.tex
\def\version{September 20, 2017}

\documentclass[12pt]{article}
\def\macrosPb{}
\def\macrosHarxiv{}

\input{macros}
\usepackage{url}

\newcommand{\chiL}{{\vartheta}}

\renewcommand{\DVa}{t}
\renewcommand\footnotemark{}

\newcommand{\multia}{a}
\DeclareMathOperator{\PT}{PT}
\DeclareMathOperator{\loc}{{\rm loc}}

\title {
   Critical two-point function for long-range
   \\
   $O(n)$ models below the upper critical dimension
}

\author{
    Martin Lohmann \and Gordon Slade
    \and Benjamin C.\ Wallace\thanks{Department of Mathematics,
        University of British Columbia,
        Vancouver, BC, Canada V6T 1Z2.
        E-mail:  {\tt marlohmann@math.ubc.ca}, {\tt slade@math.ubc.ca}, {\tt bwallace@ist.ac.at}.}}

\date\version

\newcommand{\Ncallab}{c}

\renewcommand{\chicCov}{{\vartheta}}

\newcommand{\gLfix}{\bar{s}}

\newcommand{\CRG}{C_{\rm RG}}

\begin{document}

\maketitle

\begin{abstract}
    We consider the $n$-component $|\varphi|^4$ lattice spin model ($n \ge 1$)
    and the weakly self-avoiding walk ($n=0$) on $\Z^d$, in dimensions $d=1,2,3$.
    We study long-range models based on the fractional Laplacian,
    with spin-spin interactions or walk step probabilities decaying
    with distance $r$ as $r^{-(d+\alpha)}$ with $\alpha \in (0,2)$.
    The upper critical dimension is $d_c=2\alpha$.  For $\epsilon >0$,
    and $\alpha = \frac 12 (d+\epsilon)$, the dimension $d=d_c-\epsilon$ is below
    the upper critical dimension.
    For small $\epsilon$, weak coupling, and all integers $n \ge 0$, we
    prove that the two-point function at the critical point
    decays with distance as $r^{-(d-\alpha)}$.
    This ``sticking'' of the critical exponent at its mean-field value was first
    predicted in the physics literature
    in 1972.
    Our proof is based on a rigorous renormalisation group method.  The
    treatment of observables differs from that used
    in recent work on the
    nearest-neighbour 4-dimensional
    case, via our use of a cluster
    expansion.
\end{abstract}

\ifdefined\macrosS
\fi

\section{Introduction and main result}

Broadly speaking, the mathematical understanding of critical phenomena
for spin systems has
progressed in dimension $d=2$, where exact solutions and SLE are important tools;
in dimensions $d>4$, where infrared bounds and the lace expansion are useful;
and in dimension $d=4$, where renormalisation group (RG) methods have been applied.
The physically most important case of
$d=3$ is more difficult, and mathematical methods are scarce.

In the physics literature, the $\epsilon$-expansion was introduced to study
non-integer dimensions slightly below $d=4$.  An alternate approach is to
consider long-range models, which change the upper critical dimension from $d_c=4$
to a lower value $d_c=2\alpha$ with $\alpha\in (0,2)$.  By choosing $d=1,2,3$
and $\alpha = \frac 12 (d+\epsilon)$ with small $\epsilon$, it is possible to
study integer dimension $d$ which is slightly below the upper critical dimension
$2\alpha = d+\epsilon$.  In this paper, we consider $n$-component spins and the weakly
self-avoiding walk in this long-range context, and prove that the critical two-point function
has mean-field decay $r^{-(d-\alpha)}$ also \emph{below} the upper critical dimension.
Our method involves a RG analysis in the vicinity of a non-Gaussian
fixed point.

\subsection{Introduction}

We consider long-range $O(n)$ models on $\Zd$ for integers $n \ge 0$ and dimensions $d=1,2,3$.
The case $n=0$ is the continuous-time weakly self-avoiding walk, and the case $n \ge 1$
is the $n$-component $|\varphi|^4$ lattice spin model.
For $n=0$ the underlying random walk model
takes steps of length $r$ with probabilities decaying as $r^{-(d+\alpha)}$
with $\alpha \in (0,2)$,
and for $n \ge 1$ the spin-spin interaction in the Hamiltonian has that same decay.
More precisely, the models are based on the fractional Laplacian $(-\Delta)^{\alpha/2}$,
whose kernel decays at large distance as $r^{-(d+\alpha)}$.

The upper critical dimension is predicted to be $d_c=2\alpha$ for all $n \ge 0$.
Thus, for $\alpha < \frac{d}{2}$, mean-field behaviour is predicted; this has been proved
for self-avoiding walk, for the Ising model, for the 1-component $\varphi^4$ model,
and for other models \cite{AF86,CS15,Heyd11,HHS08}.
In the physics literature,
it is observed that below the upper critical dimension the critical two-point function
continues to exhibit the mean-field decay $r^{-(d-\alpha)}$
for $\alpha \in (\frac d2,2-\eta)$,
and then crosses over to $r^{-(d-2+\eta)}$ decay for $\alpha \in (2-\eta,2)$.
Here $\eta$ is the exponent for the nearest-neighbour model; for $n=1$ this is
$\eta = \frac 14$ for $d=2$ \cite{Wu66}, and a recent estimate
for $d=3$ is $\eta = 0.03631(3)$ \cite{EPPRSV14}.
The earliest
paper to elucidate the critical behaviour of long-range models
is \cite{FMN72}, with \cite{SYI72} roughly contemporaneous and \cite{Sak73} providing
further development.
A very recent paper which analyses the crossover for the
two-point function in detail for $n=1$ is \cite{BRRZ17}.
At the crossover, when $\alpha=\alpha_*=2-\eta$,
a logarithmic correction is predicted, with overall decay $\frac{1}{r^{d-\alpha_*}}\frac{1}{\log r}$
\cite{BPR-T14,BRRZ17}.
The relationship with conformal invariance is explained in \cite{PRRZ16}.

Let $n = 0,1,2,\ldots$; $d=1,2,3$; and $\alpha = \frac 12 (d+\epsilon)$.
We use a rigorous RG argument to
prove that for small $\epsilon >0$, the critical two-point function has decay
$r^{-(d-\alpha)}$.  This proves the ``sticking'' of the critical exponent
at its mean-field value, for $\alpha$ slightly above $\frac d2$, or equivalently,
for $d$ slightly below the upper critical dimension $d_c=2\alpha$.
Our proof extends recent results and methods used to study the $\epsilon$-expansion
for the critical exponents for the susceptibility and specific heat of the long-range models
\cite{Slad17}.  It also relies on results and techniques developed to study related problems
for the 4-dimensional nearest-neighbour models \cite{BS-rg-step,BBS-saw4,ST-phi4}.
However, our treatment of observables differs from that used in the 4-dimensional case,
via our application of a cluster expansion.

Earlier mathematical work which applies RG methods to long-range
models includes the construction of global RG trajectories
for $n=0$ and $d=3$ \cite{MS08},
and for a continuum version of the $n=1$ model in
\cite{BMS03,Abde07}.  These references do not study critical exponents.  The exponents
for critical correlations in a certain hierarchical version of the model, for $d=3$ and $n=1$,
are computed in \cite{ACG13}.
For a closely related continuum model with $n=1$ in dimensions $d=2,3$, a proof of the
``sticking'' of the critical exponent for the critical two-point function
was announced in a 2013 lecture
\cite{Mitt13}.

\subsection{Fractional Laplacian}

The models we study are defined in terms of the fractional Laplacian.
We now define the fractional Laplacian and list some of its properties.
Further details can be found in \cite[Sections~\ref{alpha-sec:fL}--\ref{alpha-sec:frd}]{Slad17}.

Let $d \ge 1$ and $\alpha \in (0,2)$.  We write $|x|$ for the Euclidean norm of
$x\in \Z^d$.
Let $J$ be the $\Zd \times \Zd$ matrix with $J_{xy}=1$ if $|x-y|=1$,
and otherwise $J_{xy}=0$.  Let $I$ denote the identity matrix.  The lattice
Laplacian on $\Zd$ is
$\Delta = J-2dI$.
For $k = (k_1,\ldots,k_d) \in [-\pi,\pi]^d$, let
\begin{equation}
\lbeq{lambdak}
        \lambda(k) = 4 \sum_{j=1}^d \sin^2 (k_j/2) = 2 \sum_{j=1}^d (1-\cos k_j)
        .
\end{equation}
The matrix element $-\Delta_{x,y}$ can be written as the Fourier integral
\begin{equation}
        -\Delta_{x,y} = \frac{1}{(2\pi)^d}\int_{[-\pi,\pi]^d} \lambda(k) e^{ik\cdot (x-y)}dk.
\end{equation}
The fractional Laplacian is the matrix $(-\Delta)^{\alpha/2}$ defined by
\begin{equation}
\lbeq{QbetaFT}
        (-\Delta)^{\alpha/2}_{x,y}
        =
        \frac{1}{(2\pi)^d}\int_{[-\pi,\pi]^d} \lambda(k)^{\alpha/2} e^{ik\cdot ( x-y)}dk.
\end{equation}

For $|x-y| \to \infty$, the fractional Laplacian decays as
\begin{equation}
\lbeq{fracLapdecay}
        -(-\Delta)^{\alpha/2}_{x,y} \asymp |x-y|^{-(d+\alpha)}
\end{equation}
(see \cite[Lemma~\ref{alpha-lem:fracLapdecay}]{Slad17},
or \cite[Theorem~5.3]{BCT15} for a more precise
and more general statement).
Here, and in the following, we write $a \asymp b$ to denote the existence of $c>0$
such that $c^{-1} b \leq a \leq c b$.
For $d \ge 1$, $\alpha \in (0,2\wedge d)$, $\bar m^2>0$,
 $m^2 \in [0,\bar m^2]$, and $x \neq 0$, the resolvent obeys
\label{lem:covbd}
\begin{equation}
\lbeq{resolventbd}
                ((-\Delta)^{\alpha/2}+m^2)^{-1}_{0,x}
                \le
                c \frac{1}{|x|^{d-\alpha}} \frac{1}{1+m^4|x|^{2\alpha}},
\end{equation}
with $c$ depending on $d,\alpha,\bar m^2$ (see \cite[Lemma~\ref{alpha-lem:covbd}]{Slad17}).
For $m^2=0$, an asymptotic formula
\begin{equation}
\lbeq{resolventasy}
        ((-\Delta)^{\alpha/2})_{0,x}^{-1}\sim c_{d,\alpha} \frac{1}{|x|^{d-\alpha}}
\end{equation}
is proven in \cite[Theorem~2.4]{BC15}, with precise constant $c_{d,\alpha}$.

Given integers $L,N >1$, let $\Lambda = \Lambda_N = \Z^d/L^N\Z^d$ denote the
$d$-dimensional discrete torus of side length $L^N$.
The torus fractional Laplacian is defined by
\begin{equation}
\lbeq{fracLaptorus}
        (-\Delta_{\Lambda_N})^{\alpha/2}_{x,y}
        =
         \sum_{z \in \Zd} (-\Delta )^{\alpha/2}_{x,y+zL^N}
    \quad
    (x,y\in\Lambda_N).
\end{equation}
The sum on the right-hand side of \refeq{fracLaptorus} converges, by \refeq{fracLapdecay}.

\subsection{The \texorpdfstring{$|\varphi|^4$}{phi4} model}

We first define the model on the torus $\Lambda=\Lambda_N$, as usual for spin systems.
Let $d \ge 1$ and $\alpha \in (0,2)$.
Let $n \ge 1$.  The \emph{spin field} $\varphi$ is a function $\varphi : \Lambda \to \R^n$,
denoted $x \mapsto \varphi_x$,
which we may regard as an element $\varphi \in (\R^n)^\Lambda$.
The Euclidean norm of $v =(v^1,\ldots,v^n)\in \R^n$ is $|v| = [\sum_{i=1}^n (v^i)^2]^{1/2}$,
with inner product $v \cdot w = \sum_{i=1}^n v^i w^i$.
We extend the action of the fractional Laplacian to act on the spin field component-wise,
namely $((-\Delta_\Lambda)^{\alpha/2} \varphi)_x^i = \sum_{y\in \Lambda}
(-\Delta_\Lambda)^{\alpha/2}_{x,y} \varphi^i_y$.

Given $g>0$ and $\nu \in \R$, we define the interaction $V: (\R^n)^\Lambda \to \R$ by
\begin{equation} \label{e:VdefM}
    V(\varphi)
    =
    \sum_{x\in\Lambda}
    \big(\tfrac{1}{4} g |\varphi_x|^4 + \half \nu |\varphi_x|^2
    +
    \half
     \varphi_x \cdot ((-\Delta_\Lambda)^{\alpha/2} \varphi)_x \big)
    .
\end{equation}
The \emph{partition function} is defined by
\begin{equation}
\lbeq{pf}
        Z_{g,\nu,N} = \int_{(\R^n)^{\Lambda}}   e^{-V(\varphi)} d\varphi,
\end{equation}
where $d\varphi$ is the Lebesgue measure on
$(\R^n)^{\Lambda}$.
The expectation of a random variable $F:(\R^n)^{\Lambda} \to \R$ is
\begin{equation}
    \label{e:Pdef}
    \langle F \rangle_{g,\nu,N}
    = \frac{1}{Z_{g,\nu,N}} \int_{(\R^n)^{\Lambda}} F(\varphi) e^{-V(\varphi)} d\varphi.
\end{equation}

Given lattice points $\pp,\qq$,
we define the finite- and infinite-volume \emph{two-point function} by
\begin{align}
G_{\pp,\qq, N}(g,\nu; n) & =
 \pair{\varphi_\pp^1 \varphi_\qq^1}_{g,\nu, N}
 =
 \frac{1}{n} \pair{\varphi_\pp \cdot \varphi_\qq}_{g,\nu, N},
\\
\label{e:2pt-phi4}
G_{\pp,\qq}(g,\nu; n) & = \lim_{N \to \infty} G_{\pp,\qq, N}(g,\nu; n).
\end{align}
On the left-hand side of \refeq{2pt-phi4} we have $a,b\in\Zd$, and on the right-hand
side we identify these points with elements of $\Lambda_N$
for large $N$, by regarding the vertices of $\Lambda_N$ as a
cube in $\Z^d$ (without boundaries identified) approximately
centred at the origin.
The \emph{susceptibility} is defined by
\begin{equation}
    \label{e:susceptdef}
    \chi(g, \nu; n)
    = \lim_{N \to \infty} \sum_{b\in\Lambda_N}
    G_{a,b,N}(g,\nu;n)
\end{equation}
and can be used to identify the critical point of the model. By translation invariance, $\chi$ is independent of $a$.
Existence of the infinite volume limits in \refeq{2pt-phi4}--\refeq{susceptdef},
in our context, is discussed below.

\subsection{Weakly self-avoiding walk}

Let $d \ge 1$ and $\alpha \in (0,2)$.
Let $X$ denote the continuous-time Markov chain with state space $\Zd$
and infinitesimal generator $Q=-(-\Delta_{\Zd})^{\alpha/2}$.  Verification that
$Q$ has the attributes required of a generator is given in \cite[Lemma~\ref{alpha-lem:DJ}]{Slad17}.
Let $P$ be the probability measure associated with $X$, and $E$ the corresponding expectation;
a subscript $a$ specifies $X(0)=a$.
The transition probabilities are given by
\begin{equation}
        P_a(X(t)=b) = E_a(\1_{X(t)=b}) = (e^{tQ})_{a,b}.
\end{equation}

The \emph{local time} of $X$ at
$x$ up to time $T$ is the random variable $L_T^x = \int_0^T \1_{X(t)=x} \; dt$.
The \emph{self-intersection local time} up to time $T$ is the random variable
\begin{equation} \label{e:ITdef}
    I_T
    =
    \sum_{x\in\Z^d} \big(L_T^x\big)^2
    = \int_0^T \!\! \int_0^T \1_{X(t_1) = X(t_2)} \; dt_1 \, dt_2
    .
\end{equation}
Given $g>0$, $\nu \in \R$, and $\pp,\qq \in \Zd$,
the continuous-time weakly self-avoiding walk \emph{two-point function} is defined by
the integral
\begin{equation}
    \label{e:Gwsaw}
    G_{\pp,\qq}(g,\nu;0)
    =
    \int_0^\infty
    E_{\pp} \left(
        e^{-g I_T}
        \1_{X(T)=\qq} \right)
    e^{- \nu T}
    dT,
\end{equation}
and the \emph{susceptibility} is defined by
\begin{equation}
    \label{e:suscept-def2}
    \chi(g,\nu;0) =   \sum_{b\in \Z^d}  G_{a,b}(g,\nu;0)
    =
    \int_0^\infty E_{a}(e^{-gI_T}) e^{- \nu T} dT.
\end{equation}

The labels $0$ on the left-hand sides of \refeq{Gwsaw}--\refeq{suscept-def2} reflect
the fact that the weakly self-avoiding walk corresponds
to the formal $n=0$ case of the $n$-component $|\varphi|^4$ model.
As in earlier work on the 4-dimensional case, e.g., \cite{ST-phi4,Slad17}, we treat
both cases $n \ge 1$ (spins) and $n=0$ (self-avoiding walk) simultaneously
and rigorously, via
a supersymmetric spin representation for the weakly self-avoiding walk.

\subsection{Susceptibility and critical point}
\label{sec:suscept}

Let $d=1,2,3$; $n \ge 0$; $L$ be sufficiently large;
$\epsilon >0$ be sufficiently small; and $\alpha = \frac 12 (d+\epsilon)$.
Let $\tau^{(\alpha)}$ denote the diagonal element of the Green function, i.e.,
$\tau^{(\alpha)} = ((-\Delta)^{\alpha/2})^{-1}_{00}$.
One of the main results of \cite{Slad17} is that there exists $\gLfix \asymp \epsilon$
such that, for
$g \in [\frac{63}{64}\gLfix,\frac{65}{64}\gLfix]$,
there exist $\nu_c=\nu_c(g;n)= -(n+2) \tau^{(\alpha)} g(1+O(g))$ and $C>0$ such that
for $\nu=\nu_c+t$ with $t \downarrow 0$,
\begin{equation}
\lbeq{chigam}
        C^{-1} t^{-(1 +  \frac{n+2}{n+8} \frac{\epsilon}{\alpha} -C\epsilon^2)}
        \le
        \chi (g,\nu ;n)
        \le
        C t^{-(1 +  \frac{n+2}{n+8} \frac{\epsilon}{\alpha} +C\epsilon^2)}.
\end{equation}
This is a statement that there is a critical point at $\nu=\nu_c $,
and that the critical exponent $\gamma$ exists to order $\epsilon$, with
\begin{equation}\label{e:valgamma}
        \gamma = 1 +  \frac{n+2}{n+8} \frac{\epsilon}{\alpha} +O(\epsilon^2)
        \quad\quad
        (n \ge 0)
        .
\end{equation}
It is part of the statement that for $n \ge 1$
the susceptibility is given by the infinite-volume
limit \refeq{susceptdef}, under the above hypotheses.  The
critical exponent for the specific heat
is also computed to order $\epsilon$ in \cite{Slad17}, for $n \ge 1$.

\subsection{Main result}

Our main result is the following theorem, which shows that just
below the upper critical dimension, the exponent
for the critical two-point function ``sticks'' at its mean-field value
(see \refeq{resolventasy}), as predicted
by \cite{FMN72}.
The theorem applies
for all $n \ge 0$, including the case $n=0$ of the weakly
self-avoiding walk.  The critical value $\nu_c=\nu_c(g;n)$ is the one mentioned in
Section~\ref{sec:suscept}. As part of the proof of the theorem, it is shown that
for $n \ge 1$
the infinite-volume limit \refeq{2pt-phi4} exists for $\nu=\nu_c$.

\begin{theorem}
\label{thm:2ptfcn}
Let $d=1,2,3$; $n \ge 0$; $L$ be sufficiently large;
$\epsilon >0$ be sufficiently small; and $\alpha = \frac 12 (d+\epsilon)$.
For
$g \in [\frac{63}{64}\gLfix,\frac{65}{64}\gLfix]$ the critical two-point function
obeys, as $|a-b|\to\infty$,
\begin{equation}
\lbeq{2pt-mr}
        G_{\pp,\qq}(g, \nu_c; n) =
        (1+O(\epsilon))((-\Delta)^{\alpha/2})^{-1}_{a,b}
        \asymp  \frac{1}{|\pp-\qq|^{d-\alpha}}.
\end{equation}
\end{theorem}

Note that Theorem~\ref{thm:2ptfcn} identifies
the constant in the decay of the interacting two-point function only up to
an error of order $\epsilon$.
However, the error is uniformly bounded in $a,b$,
so the power in the decay rate
takes its mean-field value, and this is true to \emph{all} orders in $\epsilon$.

\subsection{Strategy of proof}

The proof is based on a rigorous RG method developed
in a series of papers by Bauerschmidt, Brydges and Slade, where the focus is on the nearest-neighbour
models in dimension 4.  The method is adapted to the long-range setting in \cite{Slad17}.

Fix $g$ as in the statement of Theorem~\ref{thm:2ptfcn}.  In \cite{Slad17},
given small $m^2>0$, a critical value $\nu_0^c(m^2)$ is constructed, with the property
that the critical point $\nu_c$ is given by $\nu_c=\lim_{m^2 \downarrow 0}\nu_0^c(m^2)$.
Let
\begin{align}
U_0&=  \sum_{x\in\Lambda} \left( \tfrac{1}{4} g\vert\varphi_x\vert^4
+ \half \nu_0^c\vert\varphi_x\vert^2 \right)
- \sigma_a\varphi_a^1  -\sigma_b\varphi_b^1 .
\end{align}
For $x=a,b$, let $D_{\sigma_x} = \frac{\partial}{\partial \sigma_x}|_{\sigma_a=\sigma_b=0}$.
For $n\geq 1$, the two-point function obeys
\begin{align}
G_{a,b,N}(g,\nu_0^c(m^2)+m^2;n)
&= D_{\sigma_a}D_{\sigma_b} \log \Ex_C e^{-U_0},
\end{align}
where $\Ex_C $ denotes Gaussian expectation with covariance
$C = ((-\Delta_\Lambda)^{\frac\alpha2}+m^2)^{-1}$ ($m^2>0$ ensures existence of the inverse).
Thus the two-point function is interpreted as a perturbation
of a Gaussian expectation.
A similar representation is valid for the weakly self-avoiding walk, using a
Gaussian superexpectation.

Perturbation theory is performed inductively in a multi-scale fashion,
using a finite-range decomposition $C= C_1+\cdots+ C_N $, with $C_j $ of range $\sim L^j $.
This is implemented via the Gaussian convolution identity
$\Ex_C\theta = \Ex_{C_N}\theta\circ \cdots\circ \Ex_{C_1}\theta $,
where $\Ex_C\theta $ denotes Gaussian convolution.
At every step in the induction, we get a representation
\begin{align}
\Ex_{C_j}\theta\circ\cdots\circ \Ex_{C_1}\theta e^{-U_0}\approx e^{-U_j},
\label{eq:pertapprox}
\end{align}
where the polynomial
\begin{align}
U_j &=u_j\vert\Lambda\vert +\sum_{x\in\Lambda} \left( \tfrac {1}4 g_j\vert\varphi_x\vert^4
+ \tfrac{1}2\nu_j \vert\varphi_x\vert^2 \right)
- \lambda_{a,j}\sigma_a\varphi_a^1
-\lambda_{b,j}\sigma_b\varphi_b^1
-\frac12 (q_{a,j} + q_{b,j})\sigma_a\sigma_b
\end{align}
includes all Euclidean- and $O(n)$-invariant
monomials that are relevant and marginal according to the RG philosophy.
The error in this approximation is irrelevant in the RG sense and is
controlled uniformly in the volume by parametrising it as a polymer gas.
According to (\ref{eq:pertapprox}), after the final step of the induction has been performed,
we obtain
\begin{align}
G_{a,b,N}(g,\nu_0^c+m^2;n) &= D_{\sigma_a}D_{\sigma_b}\log \Ex_C e^{-U_0}
\approx
-D_{\sigma_a}D_{\sigma_b}U_N\vert_{\varphi=0}
= \frac12 (q_{a,N}+q_{b,N}).
\label{eq:approxG}
\end{align}

To control $q_{x,N}$ ($x=a,b $), we need to study the RG dynamical system
\begin{equation}
    (g_j,\nu_j,u_j,\lambda_{x,j},q_{x,j})\to (g_{j+1},\nu_{j+1},u_{j+1},\lambda_{x,j+1},q_{x,j+1}),
\end{equation}
and its non-perturbative corrections.
The initial condition is $(g_0,\nu_0,u_0,\lambda_{x,0},q_{x,0})=(g,\nu_0^c,0,1,1,0,0)$.
(In fact, the coupling constant $u_j$ does not play an important role for the two-point function.)
For $d=4$, the dynamical system has a Gaussian fixed point.
We use the
adaptation of the RG method, as developed in \cite{Slad17}, to the long-range setting
below the upper critical dimension, where the fixed point is instead non-Gaussian.
In \cite{Slad17} only the flow of  $g_j,\nu_j,u_j$ was studied and $\lambda_j,q_j$
did not appear, but the flow of $g_j,\nu_j,u_j$ remains identical when these
additional coupling constants do appear.
For the nearest-neighbour model on $\Z^4$,
the RG method was applied in \cite{BBS-saw4,ST-phi4} to prove $|\pp-\qq|^{-2}$
decay of the critical two-point function
for all $n \ge 0$.  We mainly follow the approach of \cite{BBS-saw4,ST-phi4}.
In particular, our treatment of the flow of $q_{x,j} $ remains the same and yields
\begin{equation}
\lbeq{qjapprox}
q_{x,j} \approx \lambda_{a,j_{ab}}\lambda_{b,j_{ab}}w_{j;a,b},
\end{equation}
where
$w_{j} = \sum_{k=1}^j C_k$, and where
$j_{ab} = \lfloor \log_L (2\vert a-b\vert) \rfloor$ is the \emph{coalescence scale}
defined to ensure that $C_{k;a,b} = 0$ when $k\leq j_{ab} $.
By definition of $j_{ab}$, the right-hand side of \refeq{qjapprox} is zero
for $j$ below the coalescence scale, and this remains true non-perturbatively
as well:  $q_{x,j}=0$ for scales $j \le j_{ab}$.

The flow of $\lambda_j$ was
analysed recursively for the Gaussian RG fixed point in \cite{BBS-saw4,ST-phi4},
but for the non-Gaussian fixed point in our current setting the
recursive analysis cannot be applied
due to the non-summability of remainder terms, and a different approach is needed.
Let $\bar D = \sum_{x\in \Lambda} \frac{\partial}{\partial \varphi_x^1}|_{\varphi=0}$
and
$\bar D^2 =
\sum_{x,y\in\Lambda} \frac{\partial^2}{\partial \varphi_x^1 \partial \varphi_y^1}|_{\varphi=0}$.
According to (\ref{eq:pertapprox}),
\begin{equation}
\lambda_{a,j}
=
e^{u_j\vert\Lambda\vert} \bar D D_{\sigma_a} e^{-U_j}
\approx
e^{u_j\vert\Lambda\vert}  \bar D D_{\sigma_a}  \Ex_{w_j}\theta e^{-U_0}.
\end{equation}
Let $w_j^{(1)} = \sum_{x\in\Lambda} w_{j;a,x} $, which is independent of $a$.
Using Gaussian integration by parts and translation invariance, we show in \refeq{ddzj} that
\begin{align}
\lbeq{ibp-strategy}
\bar D D_{\sigma_a}  \Ex_{w_j}\theta e^{-U_0}
&=
\Ex_{w_j} e^{-U_0|_{\sigma=0}} + \frac1{\vert\Lambda\vert}w_j^{(1)}
\bar D^2\Ex_{w_j}\theta e^{-U_0|_{\sigma=0}} .
\end{align}
By using (\ref{eq:pertapprox}) to evaluate the two terms in the above right-hand side
approximately, we thus obtain
\begin{align}
    \lambda_{a,j} \approx 1 +  w_j^{(1)}\nu_j.
\label{eq:approxlambda}
\end{align}
This relates $\lambda_j $ to the \emph{bulk} coupling constants $g_j,\nu_j $
whose flow is known from  \cite{Slad17}.
In particular, it is shown in \cite{Slad17} that $w_j^{(1)}\nu_j = O(\epsilon) $.
All of the above is carried out uniformly in $m^2$, which permits the limit
$m^2 \downarrow 0$ to be taken after the infinite-volume limit.
Since $\lim_{m^2 \downarrow 0}\lim_{N \to \infty}
w_{N:a,b} = ((-\Delta_{\Zd})^{\alpha/2})^{-1}_{a,b}
\asymp |a-b|^{-(d-\alpha)}$, all this,
together with the rigorous versions of \refeq{qjapprox} and (\ref{eq:approxG}),
implies our main result \refeq{2pt-mr}. The non-perturbative corrections to (\ref{eq:approxlambda}) due
to the irrelevant error coordinate are controlled using a cluster expansion.
This is the main
innovation in the proof of Theorem~\ref{thm:2ptfcn}.

The remainder of the paper is organised as follows.  In Section~\ref{sec:setup}, we
provide some background and definitions needed for the RG method.  In Section~\ref{sec:RGmap},
we formulate the RG map and state the main theorem which provides estimates on the RG map;
this is an adaptation of the main result of \cite{BS-rg-step} as applied to the long-range
model in \cite{Slad17}.
The main difference, compared to \cite{Slad17}, is the inclusion of observables in the RG map.
The flow of the observable coupling constant $\lambda_j$ is analysed in Section~\ref{sec:lambda}.
The flow of the observable coupling constant $q_j$ is then analysed in Section~\ref{sec:q},
where the proof of Theorem~\ref{thm:2ptfcn} is completed.

\section{Set-up for RG method}
\label{sec:setup}

In this section, we summarise
some notation and background for the RG method,
needed for the proof of Theorem~\ref{thm:2ptfcn}.
Additional details can be found in \cite{Slad17}.

\subsection{Formula for two-point function}

We begin with a formula for the two-point function
that serves as our starting point.

\subsubsection{The case \texorpdfstring{$n \ge 1$}{n > 0}}

For $n \ge 1$, we define
\begin{equation}
\lbeq{taudef}
        \tau_x = \half |\varphi_x|^2  \quad\quad (n\ge 1).
\end{equation}
Given $g>0$, $\nu \in \R$, $m^2>0$, we set
\begin{equation} \label{e:g0gnu0nu}
    g_0 = g , \quad \nu_0 =  \nu-m^2.
\end{equation}
Given $a,b \in \Lambda$, we introduce
\emph{observable fields} $\sigma_a,\sigma_b \in \R$,
and define $V_0$ and $Z_0$ by
\begin{equation}
\lbeq{V0Z0}
        V_{0}(\varphi_x) =  g_0\tau_x^2+  \nu_0\tau_x
        - \sigma_a  \varphi_a^1\1_{x=a}  - \sigma_b  \varphi_b^1\1_{x=b},
         \quad
    Z_0(\varphi) = e^{-V_0(\Lambda)}
    ,
\end{equation}
with $V_0(\Lambda) = \sum_{x\in\Lambda}V_0(\varphi_x)$.

Given a $\Lambda\times\Lambda$ covariance matrix $w$,
let $\Ex_w$ denote the Gaussian expectation with covariance $w$.
Let $C=((-\Delta_{\Lambda_N})^{\alpha/2} +m^2)^{-1}$.
By shifting part of the $|\varphi|^2$ term into the covariance,
the expectation \refeq{Pdef} can be rewritten as
\begin{equation}
\lbeq{ExF}
\pair{F}_{g,\nu,N}
    =
\frac{\Ex_C  F e^{-V^\varnothing_0(\Lambda)}}{\Ex_C e^{-V^\varnothing_0(\Lambda)}},
\end{equation}
where $V^\varnothing_0(\Lambda)$ denotes the evaluation of $V_0(\Lambda)$
at $\sigma_a = \sigma_b = 0$.
When $F$ is a monomial, it is standard to write this ratio of expectations
as a derivative of a logarithmic generating function.
Let $D_{\sigma_a}$ denote the operator
$\frac{\partial}{\partial \sigma_a}|_{\sigma_a = \sigma_b = 0}$, and similarly
for higher derivatives.
Then the two-point function is given, for $n \ge 1$, by
\begin{equation}
\label{e:corrdiff}
        G_{a,b,N}(g,\nu;n) =
        \langle \varphi_a^1 \varphi_b^1 \rangle_{g,\nu,N}
        =
        D_{\sigma_a\sigma_b}^2 \log \Ex_C e^{- V_{0} (\Lambda)}.
\end{equation}

\subsubsection{The case \texorpdfstring{$n =0$}{n = 0}}

For $n= 0$, as in several previous papers (e.g., \cite{BBS-saw4-log,BBS-saw4,ST-phi4})
we formulate the weakly self-avoiding walk
model as the infinite-volume limit of a \emph{supersymmetric}
version of the $|\varphi|^4$ model.  The supersymmetric model involves a complex boson field
$(\phi_x,\phib_x)_{x \in \Lambda}$ and
a fermion field given by the 1-forms $\psi_x = \frac{1}{\sqrt{2\pi i}}d\phi_x$,
$\psib_x = \frac{1}{\sqrt{2\pi i}}d\phib_x$.
For $n=0$, in place of \refeq{taudef}, we set
\begin{equation}
        \tau_x =  \phi_x\phib_x + \psi_x \wedge \psib_x \quad \quad (n=0),
\end{equation}
and we replace $\varphi_a^1,\varphi_b^1$ in \refeq{V0Z0} by $\phib_a,\phi_b$.

For $n=0$, a formula closely related to \refeq{corrdiff} is given, e.g., in
\cite[\eqref{phi4-e:DaDbPNn0}]{ST-phi4}, with
$\Ex_C$ in \refeq{corrdiff} replaced by the Gaussian superexpectation.
As in \cite{ST-phi4}, our formalism applies
to the supersymmetric model with only notational changes,
with $n$ interpreted as $n=0$ in formulas such as \refeq{valgamma}, and with the Gaussian expectation replaced by a
superexpectation.
For notational simplicity, we concentrate throughout the paper on the case $n \ge 1$.

\subsection{Progressive integration}
\label{sec:prog}

In our version of the RG method, the expectation $\Ex_C e^{-V_0(\Lambda)}$ of \refeq{corrdiff}
is evaluated in a
multi-scale fashion, via a finite-range decomposition of the covariance $C$.
We use the same finite-range decomposition $C=C_1+C_2+\cdots +C_{N-1}+C_{N,N}$
of the covariance $C=((-\Delta_{\Lambda_N})^{\alpha/2} +m^2)^{-1}$ that is described
and analysed in \cite[Section~\ref{alpha-sec:frd}]{Slad17}.
A closely related decomposition was first introduced in \cite{Mitt16} and subsequently
developed in \cite{Mitt17}.
The covariances $C_j$ are translation invariant, and have the \emph{finite-range} property
\begin{equation}
    \label{e:frp}
    C_{j;x,y} = 0 \quad \text{if \; $|x-y| \geq \frac{1}{2} L^j$}
    .
\end{equation}
Thus, we may regard $C_j$ either as a covariance on $\Zd$ or on $\Lambda_N$,
as long as $N>j$. Viewing the $C_j$ as covariances on $\Zd$, we also
have a decomposition of the infinite-volume covariance given by
$((-\Delta)^{\alpha/2} + m^2)^{-1} = \sum_{j=1}^\infty C_j$.
We leave implicit the dependence of the covariance $C_j$
on $m^2$.  According to \cite[\eqref{alpha-e:scaling-estimate}]{Slad17},
for $m^2$ bounded,
the covariances $C_j$ satisfy the estimates
\begin{equation}
\label{e:scaling-estimate}
|C_{j;x,y}|
    \le
c L^{-(d - \alpha) (j - 1)} (1 + m^4 L^{2\alpha (j - 1)})^{-1}.
\end{equation}

For $n \ge 1$, and for an integrable $F : (\R^n)^\Lambda \to \R$,
we define the convolution $\Ex_C\theta F$  by
\begin{equation}
\label{e:thetadef}
    (\Ex_C\theta F)(\varphi) = \Ex_C F(\varphi+\zeta),
\end{equation}
where the expectation $\Ex_C$ on the right-hand side acts on $\zeta$ and leaves $\varphi$
fixed. A similar construction is used for $n=0$ (see, e.g., \cite[Section~\ref{log-sec:ga}]{BBS-saw4-log}).
By \cite[Proposition~\ref{norm-prop:conv}]{BS-rg-norm}, the Gaussian convolution
can be evaluated as
\begin{equation}
    \label{e:progressive}
    \Ex_{C}\theta F
    =
    \big( \Ex_{C_N}\theta \circ \Ex_{C_{N-1}}\theta \circ \cdots
    \circ \Ex_{C_{1}}\theta\big) F,
\end{equation}
with an abuse of notation where $C_N$ means $C_{N,N}$.
To compute the expectation $\Ex_C e^{-V_0(\Lambda)}$ in \eqref{e:corrdiff}, we use
\eqref{e:progressive}
to evaluate $\Ex_C \theta e^{-V_0(\Lambda)}$ progressively, as follows.
We write $\Ex_{j} = \Ex_{C_j}$ and let
\begin{equation}
\label{e:Z0def}
    Z_{j+1} = \Ex_{j+1}\theta Z_j \quad\quad
    (0 \le j<N),
\end{equation}
with $Z_0 = e^{-V_0(\Lambda)}$ as in \eqref{e:V0Z0}.
By \eqref{e:progressive}, $\Ex_CZ_0$ is obtained by
setting $\varphi =0$
in
\begin{equation}
\label{e:ZN}
        Z_N = \Ex_C \theta Z_0.
\end{equation}
This leads us to study the recursion $Z_j \mapsto Z_{j+1}$.

\subsection{Function space}

The observable fields $\sigma_a,\sigma_b$ are needed only for the purpose of evaluating
the second derivative in \refeq{corrdiff}.  Therefore, dependence on the observable
fields which is higher order than quadratic plays no role.  We make use of this
by defining the function space $\Ncal$ as explained below.  We also define the $T_\varphi$
seminorm on $\Ncal$.  These definitions are as in, e.g., \cite{BS-rg-loc,ST-phi4}.
We focus on the case $n\geq 1 $; the modifications needed for $n=0$ are as in, e.g., \cite{BS-rg-IE}.

\subsubsection{The space \texorpdfstring{$\Ncal$}{N}}
\label{sec:Ncal}

Given $p_\Ncal > 0$, let
\begin{equation}
\Ncal^\varnothing = C^{p_\Ncal}((\R^n)^\Lambda).
\end{equation}
As in \cite[Section~\ref{alpha-sec:tfnorm}]{Slad17}, we fix any $p_\Ncal \ge 10$.
For $n = 0$, $\Ncal^\varnothing$ is instead a space of even differential forms
with $p_\Ncal$-times differentiable coefficients.

In order to treat functions of the observable fields $\sigma_a,\sigma_b$,
we define an extension $\Ncal$
of $\Ncal^\varnothing$ exactly as in
\cite[Section~\ref{phi4-sec:phi4observables_representation}]{ST-phi4}.
Namely, let $\Ncal'$ be the space of real-valued functions
of $\varphi,\sigma_a,\sigma_b$
which are $C^{p_\Ncal}$ in $\varphi$ and $C^\infty$ in $\sigma_a,\sigma_b$.
An ideal $\Ical$ in $\Ncal'$ is formed by those
elements of $\Ncal'$ whose formal power series expansion in the observable fields
to order $1,\sigma_a, \sigma_b,\sigma_a\sigma_b$ is equal to zero.
We define $\Ncal$ as the quotient algebra $\Ncal =\Ncal'/\Ical$.
Then $\Ncal$ has a direct sum decomposition
\begin{equation}
\label{e:Ncaldecomp4}
\Ncal = \Ncal'/\Ical = \Ncal^\varnothing \oplus \Ncal^a \oplus \Ncal^b \oplus \Ncal^{ab},
\end{equation}
where elements of $\Ncal^a, \Ncal^b , \Ncal^{ab}$
are given by elements of $\Ncal^\varnothing$ multiplied
by $\sigma_a$, by $\sigma_b$, and by $\sigma_a\sigma_b$ respectively.
Thus, elements of $\Ncal$ can be
identified with polynomials over $\Ncal^\varnothing$ in the observable
fields with terms only of
order $1,\sigma_a, \sigma_b,\sigma_a\sigma_b$,
i.e.\ $F \in \Ncal$ can be written as
\begin{equation}
\label{e:Fdecomp}
    F = F_\varnothing + \sigma_a F_a + \sigma_b F_b +\sigma_a \sigma_b F_{ab}
\end{equation}
with $F_\Ncallab \in \Ncal^\varnothing$ for each $\Ncallab \in \{\varnothing, a,b,ab\}$.
There are natural projections $\pi_\Ncallab : \Ncal \to \Ncal^\Ncallab$ defined for such $F$
by $\pi_\varnothing F = F_\varnothing$,
$\pi_a F = \sigma_a F_a$, $\pi_b F = \sigma_b F_b$,
and $\pi_{ab} F = \sigma_a\sigma_b F_{ab}$.  We set $\pi_*=1-\pi_\varnothing$.
The expectation $\Ex_C$ acts term-by-term on $F \in \Ncal$,
namely $\Ex_C F =
\Ex_C F_\varnothing + \sigma_a \Ex_C F_a + \sigma_b \Ex_C F_b
+\sigma_a \sigma_b \Ex_C F_{ab}$ for $F$ as in \eqref{e:Fdecomp}.

\subsubsection{Seminorms}
\label{sec:Tphi}

A family of seminorms is used to control the size of elements of $\Ncal$.
Let $\Lambda^*$ denote the set of sequences of any finite length (including length $0$),
composed of elements of $\Lambda \times \{ 1, \ldots, n \}$.
Let $\varphi\in(\R^n)^\Lambda$ be a field, and let $F \in \Ncal^\varnothing$. Given
$\vec x = ((x_1, i_1), \ldots, (x_p, i_p)) \in \Lambda^*$, we write $|\vec x| = p$ and let
\begin{equation}
F_{\vec x}(\varphi)
    =
\ddp{^p F(\varphi)}{\varphi^{i_1}_{x_1} \ldots \partial\varphi^{i_p}_{x_p}}.
\end{equation}
A \emph{test function} $g$ is a mapping $g:\Lambda^* \to \R$, written
$\vec x \mapsto g_{\vec x}$.
We define the \emph{$\varphi$-pairing} of $F$ with a test function $g$ by
\begin{equation}
\label{e:pairing}
\langle F, g \rangle_\varphi
    =
\sum_{|\vec x| \le p_\Ncal}
\frac{1}{|\vec x|!} F_{\vec x}(\varphi) g_{\vec x}.
\end{equation}
Given a parameter $\h_j > 0$, a scale-dependent norm $\|g\|_\Phi = \|g\|_{\Phi_j(\h_j)}$
is defined on test functions in \cite[\eqref{alpha-e:Phinorm}]{Slad17}. The $\Phi=\Phi_j(\h_j)$
norm controls the size of a
test function and its discrete gradients up to order $p_\Phi=4$, but its precise
definition is immaterial for the present discussion. With $B_\Phi(1)$
the
unit ball in $\Phi$, we define the $T_\varphi = T_{\varphi,j}(\h_j)$
seminorm on $\Ncal^\varnothing$ by
\begin{equation}
\|F\|_{T_\varphi}
    =
\sup_{g\in  B_\Phi(1)}
|\langle F, g\rangle_\varphi|.
\end{equation}
Given an additional parameter $\h_\sigma = \h_{\sigma,j}$, we extend this definition to all of
$\Ncal$ exactly as in \cite{BS-rg-loc}, i.e., the seminorm
of $F$ of the form \eqref{e:Fdecomp} is defined to be
\begin{equation}
\lbeq{Tphiobs}
\|F\|_{T_\varphi}
        =
\|F_\varnothing\|_{T_\varphi}
        +
\h_\sigma (\|F_a\|_{T_\varphi} + \|F_b\|_{T_\varphi})
        +
\h_\sigma^2 \|F_{ab}\|_{T_\varphi}.
\end{equation}

\subsection{Blocks, polymers and scales}\label{sec:blockspolymers}

\subsubsection{Blocks and polymers}

The finite-range covariance decomposition is well-suited to a block decomposition
of the torus $\Lambda_N$ of period $L^N$ into disjoint blocks of side $L^j$, for
scales $0 \le j \le N$. This decomposition is an important ingredient in our choice of the coordinates in which we represent the RG map.
We now describe it in detail, along with a number of useful related definitions,
as in \cite{BS-rg-step}.

We partition the torus $\Lambda_N$, which has period $L^N$, into disjoint \emph{$j$-blocks} of side $L^j$
($j\le N$). Each $j$-block is a translate of the block $\{ x \in \Lambda : 0 \le x_i < L^j, i = 1, \ldots, d\}$.
We denote the collection of $j$-blocks by $\Bcal_j$.

A \emph{$j$-polymer} is any (possibly empty) union of $j$-blocks, and
$\Pcal_j$ denotes the set of $j$-polymers.
Given $X \in \Pcal_j$, we denote by $\Bcal_j(X)$ the set
of $j$-blocks in $X$,
and denote by $\Pcal_j(X)$  the set
of $j$-polymers in $X$.
A nonempty polymer $X$ is \emph{connected} if for any $x, x' \in X$, there is
a sequence $x = x_0, \ldots, x_n = x' \in X$ with $\|x_{i+1} - x_i\|_\infty = 1$
for $i = 0, \ldots, n - 1$. Let $\Ccal_j$ denote the set of connected $j$-polymers
and, for any $X\in\Pcal_j$, let ${\rm Comp}_j(X) \subset \Ccal_j(X)$ be the set of
connected components of $X$.
The empty set $\varnothing$ is not in $\Ccal_j$.

We say
that two polymers $X,Y$ \emph{do not touch} if $\min\{\|x-y\|_\infty :
x \in X, y \in Y\} >1$.
We call a connected polymer $X\in\Ccal_j$ a \emph{small set} if it consists of at most
$2^d$ $j$-blocks, and write $\Scal_j$ for the collection of small sets in $\Ccal_j$.
The \emph{small-set neighbourhood} of a polymer $X\in \Pcal_j$
is $X^\Box = \cup_{Y \in \Scal_j(X): X \cap Y \neq \varnothing}Y$.

For $F_1, F_2 : \Pcal_j \to \Ncal$, we define the scale-$j$ \emph{circle product}
$F_1 \circ F_2 : \Pcal_j \to \Ncal$ by
\begin{equation}
(F_1 \circ F_2)(Y) = \sum_{X\in\Pcal_j(Y)} F_1(Y\setminus X) F_2(X)
        \quad\quad
        (Y \in \Pcal_j).
\end{equation}
We only consider maps $F:\Pcal_j \to \Ncal$ with the property that $F(\varnothing)=1$.
The identity element for the circle product is the map
$\1_\varnothing: \Pcal_j \to \Ncal$ defined by
$\1_\varnothing(\varnothing) =1$,
$\1_\varnothing(X) =0$ if $X \ne \varnothing$.

\subsubsection{Mass and coalescence scales}

Two scales play an important role for the nature of
the RG recursion \refeq{Z0def}. We define the \emph{mass scale} $j_m$ by
\begin{equation}
\lbeq{jmdef}
        j_m = \lceil f_m \rceil, \quad\quad
        f_m = 1+ \frac{1}{\alpha} \log_L m^{-2}.
\end{equation}
By definition,
$j_m$ is the smallest scale for which $m^2L^{\alpha (j_m-1)} \ge 1$.
The mass scale is the scale beyond which the mass $m^2$ plays a significant
helpful role in the decay of the covariance $C_j$.  Indeed, by
\eqref{e:scaling-estimate} and the elementary inequality
(with notation $x_+=\max\{x,0\}$)
\begin{equation}
(1 + m^4 L^{2\alpha (j - 1)})^{-1}
    \le
L^{-2\alpha (j - j_m)_+},
\end{equation}
we have
\begin{equation}
\label{e:scaling-estimate-jm}
|C_{j;x,y}|
    \le
c L^{-(d - \alpha)(j - 1) - 2\alpha (j - j_m)_+}.
\end{equation}

We also define the \emph{coalescence scale} $j_{ab}$ by
\begin{equation}
     \label{e:jabdef}
        j_{\pp \qq}
        =
        \big\lfloor
     \log_{L} (2 |\pp - \qq|)
     \big\rfloor
     .
\end{equation}
By definition, $j_{ab}$ is the unique integer such that
\begin{equation}
\label{e:jabbds}
        \tfrac 12 L^{j_{ab}} \le |a-b| < \tfrac 12 L^{j_{ab}+1}.
\end{equation}
By \eqref{e:frp}, $C_{j;a,b}= 0$ for all $j\le j_{ab}$, and hence
\begin{equation}
    ((-\Delta)^{\alpha/2} +m^2)^{-1}_{a,b}
    = \sum_{j=1}^\infty C_{j;a,b}
    = \sum_{j=j_{ab}+1}^\infty C_{j;a,b}.
\end{equation}

Ultimately, we take the limit $m^2 \downarrow 0$ before considering large $|a-b|$,
so we can and do assume that $j_m > j_{ab}$.

\subsection{Localisation operator \texorpdfstring{$\Loc$}{Loc}}

We use the operator $\Loc$ defined and analysed in
\cite{BS-rg-loc}, to extract a local polynomial
from an element $F\in\Ncal$.  For appropriate $X \subset \Lambda$,
the local polynomial $\Loc_X F$
extracts the
parts of $F$ that are relevant and marginal in the RG sense.

\subsubsection{Local polynomials}

The range of the operator $\Loc_X$ is a certain vector space
$\Ucal(X)$ of local polynomials in the field.
We now define this vector space, taking into account that
the elements $F \in \Ncal$ to which $\Loc_X$ will be applied obey Euclidean covariance
and $O(n)$ invariance on $\Ncal^\varnothing$.

Given \emph{bulk coupling constants} $g,\nu, u \in \R$; $a,b \in \Lambda$;
\emph{observable fields} $\sigma_a,\sigma_b \in \R$;
and the \emph{observable coupling constants}
$\lambda_a,\lambda_b,q_a,q_b \in \R$; let
\begin{align}
\label{e:Unulldef}
        U^\varnothing(\varphi_x) &= g\tau_x^2 + \nu\tau_x +  u
        ,
        \\
\label{e:Udef}
        U(\varphi_x) &= U^\varnothing(\varphi_x)
        - \sigma_a \lambda_{a}\varphi_a^1\1_{x=a}  - \sigma_b \lambda_{b}\varphi_b^1\1_{x=b}
        - \tfrac12 ( q_a\1_{x=a} + q_b\1_{x=b}) \sigma_a\sigma_b
        .
\end{align}
The symbol $\varnothing$ denotes the bulk.
(For $n=0$, we can take $u=0$ in $U^\varnothing$
due to supersymmetry; see \cite{BBS-rg-pt}.)
For $U$ as in \refeq{Udef} and $X \subset \Lambda$, we write
\begin{equation}
    \label{e:UX}
    U(X) = \sum_{x \in X}U(\varphi_x)
    .
\end{equation}
Let $\Ucal(X)$ denote the space of polynomials of the form \eqref{e:UX}.
Let $\Vcal(X) \subset \Ucal(X)$ be the subspace
for which $u = q_a = q_b = 0$.
Note that $V_0$ of \refeq{V0Z0} obeys $V_0(X) \in \Vcal(X)$ with
$\lambda_{a}=\lambda_{b}=1$.

\subsubsection{Definition of Loc}

To define $\Loc$, we must first define a set of polynomial test functions, as in
\cite[Section~\ref{loc-sec:oploc}]{BS-rg-loc}.
Let $p>0$, let $\multia = (\multia_1, \ldots, \multia_p)$ with each $\multia_r \in \N_0^d$,
and let $k=(k_1,\ldots,k_p)$ with each $k_r \in \{1,\ldots, n\}$.
Let $\Lambda' \subset \Lambda$ be a coordinate patch as defined in
\cite[Section~\ref{loc-sec:oploc}]{BS-rg-loc} (e.g., $\Lambda'$ can
be any small set as defined in Section~\ref{sec:blockspolymers}).
Recall the set $\Lambda^*$ of sequences, defined in Section~\ref{sec:Tphi}.
We define a test function $q^{\multia,k}$, supported on sequences
$\vec x = ((x_1, i_1), \ldots, (x_p, i_p)) \in\Lambda^*$ with each
$x_r \in \Lambda'$, by
\begin{equation}
\label{e:poly-testfcn}
    q^{\multia,k}_{\vec x}
    =
    \prod_{r=1}^p \delta_{i_r,k_r} x_r^{\multia_r}
    =
    \prod_{r=1}^p \delta_{i_r,k_r} \prod_{l=1}^d x_{r,l}^{\multia_{r,l}}
    .
\end{equation}
We include the case $p=0$ by interpreting \eqref{e:poly-testfcn} as the constant number $1$
in this case.
The role of the coordinate patch, which cannot ``wrap around'' the torus, is to permit
polynomial test functions such as \eqref{e:poly-testfcn} to be well-defined.
We define the \emph{field dimension}
$[\varphi] =\frac{d-\alpha}{2}$.
The \emph{dimension} of the test function $q^{\multia,k}$ is defined to equal
$p[\varphi] + \sum_{r=1}^p \sum_{l=1}^d \multia_{r,l}$.
Given $d_+\ge 0$, we let $\Pi = \Pi^{d_+}[\Lambda'] \subset \Phi$ denote the span of all test
functions $q^{\multia,k}$ of dimension at most $d_+$.

Let $\Ncal^\varnothing(\Lambda')$ denote the space of functionals of the field that
only depend on field values at points in $\Lambda'$.
By \cite[Proposition~\ref{loc-prop:LTsymexists}]{BS-rg-loc}, for $X \subset \Lambda'$,
there is a unique operator $\loc_X : \Ncal^\varnothing(\Lambda') \to \Vcal(X)$
(independent of the choice of $\Lambda'$) such that
\begin{equation}
\langle F, g \rangle_0
    =
\langle \loc_X F, g \rangle_0
\text{ for all } g\in\Pi,
\end{equation}
with the pairing given by \refeq{pairing} with $\varphi=0$.
To extend $\loc$ from $\Ncal^\varnothing$ to $\Ncal$, suppose we are given $d_+^\Ncallab$ for
$\Ncallab = \varnothing, a, b, ab$. We let $\LTbar^\Ncallab_X$ denote the operator
$\LTbar_X$ with $d_+ = d_+^\Ncallab$ and,
for $F$ as in \eqref{e:Fdecomp}, define
\begin{equation}
\Loc_Y F
    =
\LTbar^\varnothing_X F_\varnothing
    +
\sigma_a \LTbar^a_X F_a
    +
\sigma_b \LTbar^b_X F_b
    +
\sigma_a\sigma_b \LTbar^{ab}_X F_{ab}.
\end{equation}
It remains to specify the $d_+^\Ncallab$.

In \cite{Slad17}, for scales below the mass scale
the range of the restriction of $\Loc$ to $\Ncal^\varnothing$ is
specified as the span of $\{1,\tau,\tau^2\}$.  This corresponds to the choice
$d_+^\varnothing = d$ (using symmetry considerations to disregard non-symmetric monomials).
Above the mass scale, a different range for $\Loc$ is used, namely $d_+^\varnothing=d-\alpha$,
due to the enhanced decay of the covariance decomposition
(see \cite[Section~\ref{alpha-sec:Loc}]{Slad17}).

As in \cite[Section~\ref{pt-sec:loc-specs}]{BBS-rg-pt}, we set
$d_+^a=d_+^b= [\varphi]$ when $\Loc$ acts at
scales strictly less than $j_{ab}$,
and set $d_+^a=d_+^b=0$ for larger scales. We always take $d_+^{ab}=0$.

The following elementary lemma will be useful.
Let $\1^{1}$ denote the constant test function supported on sequences
of length $1$ and defined by $\1^1_{(x,i)} =\delta_{i,1}$.
Likewise, let $\1^{2}$ denote the constant test function supported on
sequences of length $2$ and defined by
$\1^2_{((x_1,i_1),(x_2,i_2))}=\delta_{i_1,1}\delta_{i_2,1}$.
Note that $\1^1, \1^2$ are each of the form \eqref{e:poly-testfcn},
with respective dimensions $[\varphi]$ and $2[\varphi]$.

\begin{lemma}
\label{lem:Loc}
Given $x\in\Lambda$ and a coordinate patch $\Lambda' \ni x$, suppose that
$F \in \Ncal^\varnothing(\Lambda')$.
For $j<N$ and $m=1,2$,
\begin{align}
\pair{(1-\LT_x)F,\1^{m}}_0 = 0.
\end{align}
Moreover, if $j < j_{ab}$, then for $\Ncallab = a, b$,
\begin{equation}
\pair{(1-\LTbar^\Ncallab_x)F,\1^{1}}_0 = 0.
\end{equation}
\end{lemma}

\begin{proof}
The first statement is an immediate consequence of the definition of $\LT$
together with the fact that the test function $\1^m$ has dimension
$m[\varphi] \le d-\alpha \le d_+^\varnothing$ for $m = 1, 2$.
The second statement follows similarly using the fact that
the dimension $[\varphi]$ of $\1^1$ is equal to $d_+^a = d_+^b$ if $j < j_{ab}$.
\end{proof}

\section{RG map}
\label{sec:RGmap}

In the absence of observables, i.e., with $\sigma_a=\sigma_b=0$,
the RG map for the long-range models is constructed
and bounded in \cite{Slad17} using the main theorem of \cite{BS-rg-step}.
The result is given in \cite[Theorem~\ref{alpha-thm:step-mr}]{Slad17}.
The extension of this construction to the case of nonzero
observable fields $\sigma_a,\sigma_b$
follows a similar route as in the 4-dimensional nearest neighbour case in \cite{BBS-saw4,ST-phi4},
as we now explain.  The coordinates for the RG map are discussed in Section~\ref{sec:RGcoords},
the domain of the RG map is discussed in Section~\ref{sec:RGdom}, and the main estimates
for the RG map are given in Theorem~\ref{thm:step-mr}.
These estimates, combined with a new estimate derived from a cluster expansion,
are used in Sections~\ref{sec:lambda}--\ref{sec:qflow} to control the flow
generated by the RG map.

\subsection{RG coordinates}
\label{sec:RGcoords}

The RG
map will be defined so as to express the sequence $Z_j$ defined by \refeq{Z0def}
 as
\begin{equation}
\lbeq{zzeta}
    Z_j = e^{\zeta_j} (I_j \circ K_j)(\Lambda),
\end{equation}
for a real sequence $\zeta_j$ and sequences of maps $I_j : \Pcal_j \to \Ncal$
and $K_j : \Pcal_j \to \Ncal$.
The \emph{perturbative coordinate} $I_j$ is an explicit function of $V_j \in \Vcal$,
and
\begin{equation}
\lbeq{zetadef}
    \zeta_j = -u_j |\Lambda| + \tfrac12 (q_{a,j} + q_{b,j}) \sigma_a\sigma_b.
\end{equation}
The \emph{nonperturbative coordinate} $K_j$ is discussed in detail below.
By \eqref{e:V0Z0}, \refeq{zzeta} holds at scale $j = 0$ with
$u_0 = q_{x,0} = 0$, $I_0 = e^{-V_0}$, and $K_0 = \1_\varnothing$.
We sometimes write an element of $\Ucal$ as $U = (\zeta, V)$
with $V \in \Vcal$, where $\zeta$ encodes $u, q_a, q_b$.

We express the map
$Z_j \mapsto Z_{j+1}$ of \refeq{Z0def} via a map
$(V_j, K_j) \mapsto (\delta\zeta_{j+1}, V_{j+1}, K_{j+1})$,
the \emph{renormalisation group (RG) map},
in such a manner that
\begin{equation}
\lbeq{EIK}
\Ex_{j+1}\theta (I_j \circ K_j)(\Lambda)
    =
e^{\delta\zeta_{j+1}} (I_{j+1} \circ K_{j+1})(\Lambda)
\end{equation}
with $I_{j+1} = I_{j+1}(V_{j+1})$ and $\delta\zeta_{j+1} = \zeta_{j+1} - \zeta_j$.
This ensures that $Z_{j+1}$ has the form \eqref{e:zzeta} with
$\zeta_{j+1} = \zeta_j + \delta\zeta_{j+1}$.

\subsubsection{Perturbative coordinate}

The form of the perturbative coordinate $I_j$ is as follows.
Given a $\Lambda \times \Lambda$ matrix $w$, we define the
operator
$    \Lcal_w = \frac 12
    \sum_{u,v \in \Lambda}
    w_{u,v}
    \sum_{i=1}^n
    \frac{\partial}{\partial \varphi_{u}^i}
    \frac{\partial}{\partial \varphi_{v}^i}$.
Recall the projections defined in Section~\ref{sec:Ncal}.
Given $V',V'' \in \Vcal$, we also define
\begin{align}
\label{e:FCAB}
    F_{w}(V',V'')
    &= e^{\Lcal_{w}}
    \big(e^{-\Lcal_{w}}V'\big)
    \big(e^{-\Lcal_{w}}V'' \big) - V'V'',
\\
\label{e:Fpi}
    F_{\pi ,w}(V',V'')
    &=
    F_{w}(V',\pi_\varnothing V'')
    + F_{w}(\pi_* V',V'').
\end{align}
For $j \ge 0$, we write the partial sums of the covariance decomposition as
\begin{equation}
\label{e:wjdef}
        w_j = \sum_{i=1}^j C_i, \qquad w_0=0.
\end{equation}
As in \cite[\eqref{pt-e:WLTF}]{BBS-rg-pt}, for $B \in \Bcal_j$ we define
\begin{equation}
\lbeq{Wdef}
    W_j(V,x)
    = \frac 12 (1-\Loc_x) F_{\pi,w_j}(V_x,V(\Lambda)),
    \quad\quad
    W_j(V,B)
    = \sum_{x \in B} W_j(V,x).
\end{equation}
The polynomial $W_j(V, B)$ in the fields is thus an explicit quadratic function of $V$.
In particular, $W_j^\varnothing$ is an even polynomial in the fields, and
$W_j$ is quadratic in the coupling constants and is irrelevant in the RG sense.
Finally, for $V\in\Vcal$, we define
$I_j = I_j(V, \cdot\,) : \Pcal_j \to \Ncal$  by
\begin{equation}
\lbeq{Idef}
I_j(V, X)
= e^{-V(X)} \prod_{B\in\Bcal_j(X)} (1 + W_j(V, B)).
\end{equation}

As in \refeq{corrdiff}, we write
 $D_{\sigma_a}=\frac{\partial}{\partial \sigma_a}|_{\sigma_a = \sigma_b = 0}$.
 We also write
$\bar D = \sum_{x\in\Lambda} \frac{\partial}{\partial \varphi_x^1}|_{\varphi=0}$
and $\bar D^2 = \sum_{x,y\in\Lambda} \frac{\partial^2}{\partial \varphi_x^1 \partial \varphi_y^1}|_{\varphi=0}$.
We will later make use of the following corollary of Lemma~\ref{lem:Loc}.

\begin{cor}
\label{cor:W0}
For $V \in \Vcal$ and $x \in \Lambda$, and with $W_j=W_j(V,x)$,
\begin{equation}
\bar D W^\varnothing_j = \bar D^2 W^\varnothing_j = 0.
\end{equation}
Moreover, if $j < j_{ab}$, then
\begin{equation}
\bar D D_{\sigma_a} W_j = 0.
\end{equation}
\end{cor}

\begin{proof}
The fact that $\bar D W^\varnothing_j = 0$ is immediate since $W_j^\varnothing$
is even in the fields.
Also, since $\bar D^2 W_j^\varnothing =
\langle W^\varnothing_j, \1^2 \rangle_0$, it follows from Lemma~\ref{lem:Loc}
that $\bar D^2 W_j^\varnothing =0$.
To see that $\bar D D_{\sigma_a} W_j = 0$, note that
$D_{\sigma_a} W_j = \tfrac12 (1 - \LTbar^a_x) F_{a}$, with
$F_a \in \Ncal^\varnothing$ the coefficient of $\sigma_a$ in $F_{\pi,w_j}(V_x, V(\Lambda))$.
Thus, by definition \eqref{e:pairing} of the
pairing, $\bar D D_{\sigma_a} W_j = \tfrac12 \langle (1 - \LTbar^a_x) F_{a}, \1^1 \rangle_0$,
which vanishes by Lemma~\ref{lem:Loc} when $j < j_{ab}$. This completes the proof.
\end{proof}

\subsubsection{Nonperturbative coordinate}

We now define the space $\Kcal_j$
of maps $K : \Pcal_j \to \Ncal$ which contains the nonperturbative RG
coordinate. With $\Ncal$ replaced by $\Ncal^\varnothing$, such a space is
defined in \cite[Definition~\ref{alpha-def:Kspace}]{Slad17}, and, as in \cite{BS-rg-step},
we extend it here
to include observables.
The symmetries (Euclidean covariance, gauge invariance, supersymmetry, and $O(n)$-invariance)
used in Definition~\ref{def:Kspace} are defined in
\cite[Section~\ref{step-sec:coordinates}]{BS-rg-step} and \cite[Section~\ref{phi4-log-sec:margrel}]{BBS-phi4-log}.
For $n \ge 1$, as a replacement for the gauge invariance which holds for $n=0$
we also introduce \emph{sign invariance}, which
is invariance under the map $(\sigma,\varphi) \mapsto (-\sigma,-\varphi)$.
Note that $V_0$ of \refeq{V0Z0} is sign invariant.  It can be verified that the property
of sign invariance is preserved by the map $K \mapsto K_+$ of \cite{BS-rg-step}.

\begin{defn}
\label{def:Kspace}
For $j<N$, let $\CKspace_{j} = \CKspace_{j} (\Lambda)$
denote the real vector space of functions
$K : \Ccal_j \to \Ncal$ with the following properties:
\begin{itemize}
\item Field Locality:
For all $X\in \Pcal_j(\Lambda)$, $K (X) \in \Ncal (X^{\Box})$. Also,
(i) $\pi_{a} K (X)=0$ unless $a\in X$,
(ii) $\pi_{b} K (X)=0$ unless $b\in X$, and
(iii) $\pi_{ab} K (X)=0$ unless $a\in X$ and $b\in X^\Box$ or vice versa,
    and $\pi_{ab} K(X)=0$ if $X\in \Scal_j$ and $j<j_{ab}$.

\item Symmetry:
(i) $\pi_\varnothing K$ is Euclidean covariant,
(ii) if $n = 0$, $\pi_\varnothing K$ is supersymmetric and $K$ is gauge invariant
    and has no constant part;
if $n \ge 1$, $\pi_\varnothing K$ is $O(n)$-invariant and $K$ is sign invariant.
\end{itemize}
Let $\Kspace_{j} = \Kspace_{j} (\Lambda)$ be the real vector space
of functions $K : \Pcal_j  \to \Ncal$ which have
the above field locality and symmetry properties, and, in addition:
\begin{itemize}
\item Component Factorisation: for all polymers $X$, $K (X) = \prod_{Y
\in {\rm Comp}( X)}K (Y)$.
\end{itemize}
\end{defn}

The nonperturbative coordinate $K_j$ appearing in \refeq{zzeta} is an element
of $\Kspace_j$.
An element of $\Kspace_{j}$ determines an element of $\CKspace_{j}$
by restriction to $X\in\Ccal_j$.  Also, an element of $\CKspace_{j}$
determines an element of $\Kspace_{j}$ by the factorisation property.
The same symbol is used for both elements related by this
correspondence.  Since the empty set is not a connected set,
$\1_{\varnothing} \in \Kspace_{j}$ becomes $0 \in \CKspace_{j}$
under this correspondence.

\subsection{Norms and RG domain}
\label{sec:RGdom}

We now specify the domain of the RG map, which requires specification of norms
on the spaces $\Vcal$ and $\CKspace$.
Without the observables fields, the norms
are discussed in \cite[Section~\ref{alpha-sec:nr}]{Slad17}.
For the nearest-neighbour 4-dimensional case,
the adaptation of the norms to include
observables is discussed,
e.g., in \cite[Section~\ref{phi4-sec:par}]{ST-phi4}.  For our current long-range setting, we need
only adjust some norm parameters, compared to \cite[Section~\ref{phi4-sec:par}]{ST-phi4}.

As in \cite[\eqref{alpha-e:sfix}]{Slad17}, the small number $\gLfix$ in Theorem~\ref{thm:2ptfcn}
is given by
\begin{equation}
\lbeq{sfix}
        \gLfix
        =\frac 1a (1-L^{-\epsilon})
        = O(\epsilon),
\end{equation}
with the constant $a$ specified in
\cite[Lemma~\ref{alpha-lem:beta-a0}]{Slad17} (not to be confused with the point $a\in \Zd$ used
for the two-point function).
Recall the mass scale $j_m$ defined in \refeq{jmdef}.
Following \cite[\eqref{alpha-e:alphapdef}--\eqref{alpha-e:elldefa}, \eqref{alpha-e:hdef}]{Slad17},
we fix
\begin{equation}
\lbeq{alphapdef}
        \alpha'
        \in (0, \half \alpha),
\end{equation}
and define the bulk parameters
\begin{align}
\lbeq{elldefa}
        \ell_j &
        = \ell_0 L^{-\frac 12 (d-\alpha)j} L^{- \frac 12 (\alpha+\alpha') (j-j_m)_+}
        =
        \begin{cases}
        \ell_0 L^{-\frac 12 (d-\alpha)j} & (j \le j_m)
        \\
        \ell_0 L^{-\frac 12 (d-\alpha)j_m}L^{-\frac 12 (d+\alpha')(j-j_m)} & (j > j_m),
        \end{cases}
        \\
        \lbeq{hdef}
        h_j &
        = \frac{1}{\gLfix^{1/4}}k_0 L^{-\frac 12 (d-\alpha)j}
        = \frac{1}{\gLfix^{1/4}}\frac{k_0 }{\ell_0} \ell_j
        \quad\quad
        (j \le j_m)
        .
\end{align}
Here $\ell_0$ can be chosen large (depending on $L$) and $k_0$ is a fixed (small) constant.
We use $\h_j$ to refer to either of the bulk parameters $\ell_j,h_j$.

Now that observables are present, the pair of parameters $\h_j$ is supplemented by
the pair
\begin{equation}\label{e:weightsigma}
        \h_{\sigma,j} =
         \ell_{j \wedge j_{ab}}^{-1}
        2^{(j - j_{ab})_+}
        \times
        \begin{cases}\gLfix & (\h = \ell)
        \\
        \gLfix^{1/4} & (\h = h).
        \end{cases}
\end{equation}
We only use $h_{\sigma,j}$ for $j \le j_m$. Recall that we assume that
the coalescence scale $j_{ab}$ is smaller than the mass scale $j_m$,
since the limit $m^2 \downarrow 0$ will be taken before considering arbitrarily large
$|a-b|$.

For $U \in \Ucal \simeq \R^7$, we define the scale-dependent norm
\begin{equation}
\lbeq{Ucalnorm}
    \|U\|_\Ucal = \max
    \Big\{
    |g| L^{\epsilon (j \wedge j_m)},
    |\nu| L^{\alpha (j \wedge j_m)},
    |u| L^{j d},
    \ell_j \ell_{\sigma,j} (|\lambda_a| \vee |\lambda_b|),
    \ell_{\sigma,j}^2 (|q_a| \vee |q_b|)
    \Big\}
    .
\end{equation}
We denote the restriction of $\|\cdot\|_\Ucal$ to $\Vcal$ by the same symbol.
Given $C_\DV >0$, we
define the domain
\begin{equation}
\lbeq{DVdef-obs}
\DV_j =
\{
    V \in \Vcal :
    \|V\|_\Ucal \le C_\DV \gLfix, \;
    g >
    C_\DV^{-1} \gLfix L^{-\epsilon (j \wedge j_m)}
\}
    \subset \Vcal
    .
\end{equation}
Note that $\DV_j$ is a domain in $\Vcal$, and as such, does not
involve the coupling constants $u$ or $q$.

A sequence $\Wcal^\varnothing_j$ of Banach spaces is defined
in terms of the $T_\varphi(\h_j)$ seminorms in \cite{Slad17} (they are
denoted $\Wcal_j$ there). We extend $\Wcal^\varnothing_j$ to a space
$\Wcal_j \subset \CKspace_j$ whose definition is the same with the exception
that the $T_\varphi$ seminorms are defined on the extended space $\Ncal$.
As in \cite[Remark~\ref{alpha-rk:chiL}]{Slad17}, we define a sequence
\begin{equation}
\lbeq{chiLdef}
\chiL_j = L^{-\tfrac14 \alpha (j - j_m)_+}.
\end{equation}
Given a parameter $\DVa > 0$, the domain of the RG map is defined by
\begin{equation}
\lbeq{domRG}
\domRG_j = \DV_j \times B_{\Wcal_j}(\DVa \chicCov_j^3 \gLfix^3),
\end{equation}
where $B_{\Wcal_j}(r)$ is the open ball of radius $r$ in the Banach space $\Wcal_j$.

\subsection{Estimates on RG map}

We now specify the RG map $(V_j, K_j) \mapsto (U_{j+1},K_{j+1})=
(\delta\zeta_{j+1}, V_{j+1}, K_{j+1})$
and state our bounds on it.
To shorten notation, we condense indices and write, e.g.,  $(V,K) $ for $(V_j, K_j)$
and $(U_+,K_+) $ for $(U_{j+1},K_{j+1})$.
The definition of the maps $U_+,K_+$ is described in a general setup in
\cite{BBS-rg-pt,BS-rg-step}, and is adapted to the long-range model
with $\sigma_a=\sigma_b=0 $ in \cite{Slad17}.
The same definitions extend to include observables.

In particular, the map $(V,K) \mapsto U_+=(\delta\zeta_+, V_+) = \PT(V) + R_+(V,K)$
is explicit and consists of a perturbative part $\PT$, incorporating
second-order perturbation theory, and a nonperturbative, third-order
error $R_+$.
The explicit map $\PT $ is the one defined in \cite{BBS-rg-pt} for $n=0$, extended in
\cite{BBS-phi4-log} to $n \ge 1$,
and used in \cite{ST-phi4} for general $n \ge 0$.
Let $\lambda$ denote $\lambda_a$ or $\lambda_b$, and let
$q$ denote $q_a$ or $q_b$. We denote
the $\lambda, q$ components of the map $\PT$ by $\lambda_\pt, \delta q_\pt$.
For $j<N$, and with $w_j$ given by \refeq{wjdef}, let
\begin{equation}
\lbeq{w1def}
    w_j^{(1)} = \sum_{y \in \Zd}w_{j;x,y}.
\end{equation}
By \cite[\eqref{alpha-e:Greeknoprime} and Lemma~\ref{alpha-lem:wlims}]{Slad17},
\begin{equation}
\label{e:w1bd}
    w_j^{(1)} = O(L^{\alpha j}).
\end{equation}
Let
\begin{equation}
\lbeq{delnuw}
    \delta[\nu w^{(1)}] = (\nu + (n+2) g C_{+;0.0})w_{+}^{(1)} - \nu w^{(1)}.
\end{equation}
Recall the definition of the coalescence scale $j_{ab}$ in \refeq{jabdef}.
Then, as in \cite[Proposition~\ref{phi4-prop:pt}]{ST-phi4}, for general $n \ge 0$ the observable part of
the map $\PT$ is the map  $V \mapsto (\lambda_\pt, \delta q_\pt)$
given by
\begin{align}
    \lambdapt
    &
    =
    \begin{cases}
    (1 - \delta[\nu w^{(1)}])\lambda & (j+1 < j_{\pp\qq})
    \\
    \lambda & (j+1 \ge j_{\pp\qq}),
    \end{cases}
\label{e:lambdapt2}
    \\
    \delta \qpt
    &
    =
    \lambda_a\lambda_b \,C_{j+1;\pp,\qq}
\label{e:qpt2}
    .
\end{align}
Note that $\lambda_\pt = \lambda$ for all scales $j \ge j_{ab}-1$, i.e.,
the flow of $\lambda$ stops evolving after scale $j_{ab}-1$.
Conversely, since $C_{j+1;\pp,\qq}=0 $ for $j+1\leq j_{ab} $,
nonzero $\delta q_{\pt}$ can occur only at scales $j \ge j_{ab}$.
The map $(V,K) \mapsto U_+$ is now defined by
\begin{equation}
\lbeq{Vhatdef}
    U_+ = \PT(\hat V) \quad \text{with} \quad
    \hat V = V -  \!\!\!\! \sum_{Y \in \Scal (\Lambda) : Y \supset B} \!\!\!\!
    \LT_{Y,B} I^{-Y} K (Y) .
\end{equation}
The localisation operator $\LT_{Y,B}$ is defined in \cite[Definition~\ref{loc-def:LTXYsym}]{BS-rg-loc}.
The higher-order correction $R_+ : \Vcal \to \Ucal$ to the
perturbative calculation is then defined by
\begin{equation}
R_+(V, K) = \PT(\hat V) - \PT(V),
\end{equation}
so that $U_+= \PT(V)+R_+(V,K)$.
We do not need the explicit form of $R_+$
and only use the bounds of Theorem~\ref{thm:step-mr} below.

The map $(V, K)\to K_+$ is also given explicitly in \cite{BS-rg-step}, but
it is complicated to write down.  Like $R_+$, this nonperturbative part of the RG map is
of order $O(\gLfix^3)$.
It is part of the statement of Theorem~\ref{thm:step-mr} below that the
formula for $K_+$ constructed in \cite{BS-rg-step} is well-defined on the domain
specified in Theorem~\ref{thm:step-mr}.  We do not need to know more here about $K_+$
than the estimates provided by Theorem~\ref{thm:step-mr}.

The RG map depends on the mass $m^2$ through its dependence on the
covariance $C_{+}$. We require continuity in the mass
in the limit $m^2\downarrow 0$,
which can only be taken after the infinite-volume limit $N\to\infty$.
Given small $\delta > 0$, we define
the mass domain for the RG map by
\begin{equation}
\lbeq{massint}
    \Iint_j = \begin{cases}
    [0,\delta] & (j<N)
    \\
    [\delta L^{-\alpha(N-1)},\delta] & (j=N).
    \end{cases}
\end{equation}
The special attention to $j=N$ is due to the fact that the final covariance $C_{N,N}$
is only defined for $m^2>0$, and it obeys good estimates for $m^2\in \Iint_N$.

The following extends \cite[Theorem~\ref{alpha-thm:step-mr}]{Slad17}
to allow for the presence of observables.
Its estimates appear
identical to
\cite[Theorem~\ref{alpha-thm:step-mr}]{Slad17}, but it is in fact an extension
since the domain and range of the RG map now include
observables in $(V,K)\in\domRG$, as well as in $R_+$ and $K_+$.
Note that the map $R_+$, which acts on $(V,K)$ with
$V \in \Vcal$, produces a polynomial in $\Ucal$ which in particular contains
the nonperturbative contributions to $\delta \zeta$.  The bound \eqref{e:Rmain-g}
on $R_+$ controls these nonperturbative contributions to $\delta \zeta$.
Note that the estimates \refeq{Rmain-g} hold for $m^2 \in \Iint_+$, but
the continuity is in the smaller interval $m^2 \in [0,L^{-\alpha j}]$.  A restriction
like this on the continuity interval is essential, because larger $m^2$ will put
$j$ above the mass scale, at which point the spaces themselves become dependent on $m^2$
through their dependence on $\ell_j$ and a continuity statement becomes meaningless.

\begin{theorem}
\label{thm:step-mr}
    Let $d =1,2,3$; $n \ge 0$; $\alpha = \frac 12 (d+\epsilon)$ and $j<N$.
    Let $C_\DV$ and $L$ be sufficiently large, and let $\epsilon$ be sufficiently small.
    There exist $\CRG>0$, $\delta >0$,
     such that,
    with the domain
    $\domRG$ defined using $\DVa =4 \CRG$, the maps
    \begin{equation}
    \lbeq{RKplusmaps}
        R_+:\domRG  \times \Iint_+  \to \Ucal,
        \quad
        K_+:\domRG  \times \Iint_+  \to \Wcal_{+}
    \end{equation}
    are analytic in $(V,K)$, provide a solution to \refeq{EIK},
    and satisfy the estimates
    \begin{align}
\label{e:Rmain-g}
        \|R_+\|_{\Ucal_+}
        & \le
        \CRG
        \chiL_+  \gLfix^{3} ,
\quad\quad
        \|K_+\|_{\Wcal_+}
        \le
        \CRG  \chiL_+^3  \gLfix^{3}.
    \end{align}
The coordinate in $R_+ $ corresponding to $\delta q_{a,j},\delta q_{b,j} $
is identically zero for $j\leq j_{ab} $, and the coordinate corresponding to
$\lambda_{a,j},\lambda_{b,j} $ is identically zero for $j \ge j_{ab}$.
In addition, $R_+,K_+$ are jointly continuous in $m^2 \in [0,L^{-\alpha j}], V,K$.
\end{theorem}

\begin{proof}
The theorem is a consequence of the main result of \cite{BS-rg-step},
which focusses on the 4-dimensional nearest-neighbour case.
For the long-range model, the appropriate modifications for the bulk part of the
RG map are discussed in \cite{Slad17}, and we assume familiarity with both the methodology
and the modifications.
In order to include observables,
only minor further
modifications are required, compared to \cite{BS-rg-IE,BS-rg-step}.

One requirement is to verify that, for $V\in \DV$,
the basic small parameters $\epsilon_V $ and $\epdV$ obey appropriate estimates
when observables are present, as in \cite[Sections~\ref{IE-sec:epVW}--\ref{IE-sec:epdV-app}]{BS-rg-IE}.
We verify this here; this verification validates our choice \refeq{weightsigma}
for the norm parameters.
(In fact, somewhat larger domains are used in \cite[Sections~\ref{IE-sec:epVW}--\ref{IE-sec:epdV-app}]{BS-rg-IE};
the main ideas are present for $V\in\DV$, which we consider here,
and the extension to the larger domains
is a matter of bookkeeping.)
A second requirement is to verify that the
``crucial contraction'' is maintained in the presence of observables,
and we also verify this here.

\medskip \noindent \emph{Bound on $\epsilon_V$.}
Let $V\in \DV$.
For $\epsilon_V$, it suffices to observe that for
$|\lambda_a| \le C_\DV \gLfix \ell^{-1}\ell_\sigma^{-1}$,
\begin{align}
\label{e:lam-stability}
        \|\lambda_a \sigma_a  \varphi_a^1\|_{T_0(\h)}
        & =
        |\lambda_a| \h_{\sigma}\h
        \le
        C_\DV \gLfix \frac{\h_{\sigma}}{\ell_\sigma} \frac{\h}{ \ell}
        =
        \begin{cases}
        C_\DV \gLfix & (\h=\ell)
        \\
        C_\DV k_0 \ell_0^{-1}  & (\h=h),
        \end{cases}
\end{align}
which implies stability on the domain $\DV$ of
\refeq{DVdef-obs}, and complements the arguments of \cite{BS-rg-step,Slad17}.

\medskip \noindent \emph{Bound on $\epdV$.}
We must also verify the analogue of \cite[Lemma~\ref{IE-lem:epdV}]{BS-rg-IE}.
To state the desired estimate, as in \cite[\eqref{alpha-e:ellhatdef}]{Slad17}
we define the norm parameter
\begin{equation}
\lbeq{ellhatdef}
\hat\ell_j
    =
\hat\ell_0 \ell_j L^{-\frac 12 (\alpha-\alpha')(j-j_m)_+},
\end{equation}
and as in \cite[\eqref{alpha-e:epdVdef}]{Slad17} we define the small parameter
\begin{align}
\lbeq{epdVdef}
\epdV &= \epdV_j(\h) =
\begin{cases}
\gLfix \chiL_j & (\h=\ell)
\\
\gLfix^{1/4} & (\h=h,\; j \le j_m)
.
\end{cases}
\end{align}
We write $U_\pt = \PT(V)$
and let
$\delta V = \theta V - U_\pt$.
Our goal then is to show that, for $V \in \DV$,
\begin{equation}
\label{e:dVbd}
\max_{B\in\Bcal} \|\delta V(B)\|_{T_0(\h\sqcup\hat\ell)}
    \le
C_{\delta V} \epdV
\end{equation}
where $C_{\delta V}$ is an $L$-dependent constant.

It is argued in \cite[Section~\ref{alpha-sec:epdV}]{Slad17}
that \eqref{e:dVbd} holds with $\delta V$ replaced by
$\delta V^\varnothing = \pi_\varnothing \delta V$.
Thus, it suffices to establish \eqref{e:dVbd} with $\delta V$ replaced by
$\delta V^* = \pi_* \delta V$. This can be done by writing
\begin{equation}
\|\delta V^*(B)\|_{T_0(\h\sqcup\hat\ell)}
    \le
\|\theta V^*(B) - V^*(B)\|_{T_0(\h\sqcup\hat\ell)}
    +
\|V^*(B) - U_\pt^*(B)\|_{T_0(\h\sqcup\hat\ell)}
\end{equation}
and applying the triangle
inequality to estimate each of the two terms on the right-hand side.

For instance, if $a \in B$, then the $\sigma_a$ term of $\theta V^*(B) - V^*(B)$
is $\lambda_a \sigma_a \zeta^1_a$.
By definition of the norm, by  \eqref{e:ellhatdef}, \eqref{e:weightsigma},
\eqref{e:epdVdef}, \refeq{chiLdef}, and by the fact that $\alpha' < \tfrac12 \alpha$,
\begin{equation}
\|\lambda_a\sigma_a\zeta^1_a\|_{T_0(\h\sqcup\hat\ell)}
    = |\lambda_a|\h_{\sigma}\hat\ell
    \le
C_\DV \gLfix  \hat\ell_0 L^{-\frac12 (\alpha - \alpha') (j - j_m)_+} \frac{\h_\sigma}{\ell_\sigma}
    \le
C_\DV \hat\ell_0 \epdV(\h).
\end{equation}
The  $\sigma_a$ term of $V^*(B) - U_\pt^*(B)$ is zero above the coalescence scale,
whereas if $j+1 < j_{ab}$ then it
is $\delta[\nu w^{(1)}] \lambda_a \sigma_a \varphi^1_a$, by \refeq{lambdapt2}.
Thus, by \eqref{e:lam-stability}, it is sufficient to show that
\begin{equation}
\lbeq{delnuw-suff}
|\delta[\nu w^{(1)}]| \le \bar s^{1/4} \chiL.
\end{equation}
By its definition in \refeq{delnuw},
\begin{equation}
\label{e:deltanuw-re}
\delta[\nu w^{(1)}]
    =
\nu\sum_x C_{+:0,x} + (n + 2) g C_{+;00} w^{(1)}_+.
\end{equation}
By \eqref{e:DVdef-obs} and \eqref{e:scaling-estimate-jm},
and the finite-range property \refeq{frp},
the first term is bounded by
\begin{equation}
O(\bar s) L^{-\alpha (j \wedge j_m)} L^{jd} L^{-(d - \alpha) j - 2 \alpha (j - j_m)_+}
    =
O(\bar s) L^{-\alpha (j - j_m)_+} ,
\end{equation}
and the second term is bounded by
\begin{equation}
O(\bar s) L^{-\epsilon (j \wedge j_m)} L^{-(d - \alpha) j - 2 \alpha (j - j_m)_+} L^{\alpha j}
    =
O(\bar s) L^{- d (j - j_m)_+}.
\end{equation}
These bounds do better than what is required by \refeq{delnuw-suff}.

For the $\sigma_a\sigma_b$ term, we can take $j \ge j_{ab}$.
The $\sigma_a \sigma_b$ term of $\theta V^* - V^*$ is always $0$ and the
coefficient of
$\sigma_a\sigma_b$ in $V^*(B) - U_\pt^*(B)$
is at most
$|C_{+;a,b} \lambda_a\lambda_b|$.
By \eqref{e:DVdef-obs}, \eqref{e:scaling-estimate-jm}, and \eqref{e:elldefa}
and the fact that $\alpha'<\frac 12 \alpha$,
\begin{equation}
\lbeq{Csigsig}
\|C_{+;a,b} \lambda_a\lambda_b \sigma_a\sigma_b\|_{T_0(\h\sqcup\hat\ell)}
    \le
O(\bar s^2)  \frac{|C_{+;a,b}|}{\ell^2} \frac{\h_\sigma^2}{\ell_\sigma^2}
    \le
O(\bar s^2) \chiL^2
\frac{\h_\sigma^2}{\ell_\sigma^2}.
\end{equation}
When $\h = h$ (hence $j < j_m$) this is
$O(\epdV(h)^2)$,
and when $\h = \ell$ it is $O(\epdV(\ell)^2)$.  This is better than what is required for
\refeq{dVbd}.

\medskip \noindent \emph{Crucial contraction.}
The adaptation of the crucial contraction to the long-range model is provided for
the bulk in \cite[Section~\ref{alpha-sec:kappabms}--\ref{alpha-sec:kappapms}]{Slad17}.
We now extend the adaptation to include observables.

Below the mass scale, the least irrelevant of the sign invariant
monomials involving the observable fields
each have two additional spin fields compared to their marginal counterparts
$\sigma_a \varphi_a^1$ and $\sigma_a\sigma_b$
(the latter occurs only above the coalescence scale), so have dimension which is larger
by $2[\varphi]=d-\alpha$.
Compared to \cite[\eqref{alpha-e:kappa}]{Slad17},
this gives rise to $\gamma = L^{-(d-\alpha)}$, and there is no factor $L^{d}$
for observables, so the gain here is proportional to $L^{-(d-\alpha)}$.  The worst $\gamma$ occurs
for $d=1$, where we have $\gamma = L^{-\frac d2 + \frac \epsilon 2}
= L^{-\frac 12 + \frac \epsilon 2}$.
This is consistent with the values of $L^d\cgam$ reported for the bulk in
\cite[\eqref{alpha-e:kappa}]{Slad17}.

Above the mass scale, we extend the discussion in \cite[Section~\ref{alpha-sec:kappapms}]{Slad17}, as follows.
For the perturbative contribution to $K$, we have already verified that we can continue
to use the $\epdV$ given by \refeq{epdVdef} when observables are present, and there is
therefore no change to \cite{Slad17} concerning this issue.  It remains to consider the
crucial contraction.

We recall and invoke our assumption that $j_{ab}<j_m$.
Now $d_+^a=0$, so the least irrelevant monomial in $\Ncal^a$ is
$\sigma_a \varphi$.
This scales as
\begin{align}
    \ell_{\sigma,j}\ell_j & = \frac{\ell_j}{\ell_{j_{ab}}}2^{j-j_{ab}}\gLfix
    =
    \frac{L^{-\frac 12 (d-\alpha)j}L^{-(\alpha+\alpha')(j-j_m)}}
    {L^{-\frac 12 (d-\alpha)j_{ab}}}2^{j-j_{ab}}
    \gLfix
    \nnb
    & \le
    L^{-\frac{1}{2} (d+ \alpha')(j-j_m)} 2^{j-j_m}\gLfix
    .
\end{align}
A change from scale $j$ to scale $j+1$ in the above right-hand side gives rise
to a factor $2L^{-\frac 12 (d+\alpha')}$.
As in
\cite[\eqref{alpha-e:kappajm}--\eqref{alpha-e:kappa-above-jm-5}]{Slad17},
the essential condition here is that
the product of this factor with $\chiL^{-3} = L^{\frac 34 \alpha}$ should be bounded
above by an inverse power of $L$.  This condition is indeed satisfied, since
\begin{equation}
    \half (d+\alpha') - \tfrac{3}{4} \alpha
    =
    \half \left( d+\alpha' - \tfrac{3}{4}(d +\epsilon) \right)
    =
    \half \left( \tfrac{1}{4} d+\alpha' - \tfrac{3}{4}\epsilon) \right) > \tfrac {1}{8} d.
\end{equation}

Similarly, the least irrelevant monomial in $\Ncal^{ab}$ that is
sign invariant is of the form $\sigma_a \sigma_b \varphi \varphi$,
and has scaling dimension
twice that considered in the previous paragraph, so twice as good.  Thus the crucial
contraction is not harmed by the presence of observables.

\medskip \noindent \emph{Estimate for $R_+$ above the mass scale.}
Finally, we consider the extension of \cite[Section~\ref{alpha-sec:Rbound}]{Slad17}
to include observables.
The observable terms have the same $T_0$ and $\Ucal$ norms:
$\|\sigma_a\varphi_a^1\|_{T_0} = \ell_\sigma \ell = \|\sigma_a\varphi_a^1\|_{\Ucal}$
and
$\|\sigma_a\sigma_b\|_{T_0} = \ell_\sigma^2 = \|\sigma_a\sigma_b\|_{\Ucal}$.
This leads to an extension to
\cite[Lemma~\ref{alpha-lem:monnormcomp}]{Slad17},
as follows.
Let
\begin{align}
    F_1
    & =
    \nu\tau + u + - \sigma_a  \varphi_a^1\1_{x=a}  - \sigma_b  \varphi_b^1\1_{x=b}
    - \tfrac12 ( q_a\1_{x=a} + q_b\1_{x=b}) \sigma_a\sigma_b
    ,
    \\
    F_2
    &=
    g\tau^2 + \nu\tau + u + - \sigma_a  \varphi_a^1\1_{x=a}  - \sigma_b  \varphi_b^1\1_{x=b}.
\end{align}
The estimates of
\cite[Lemma~\ref{alpha-lem:monnormcomp}]{Slad17}
now become
\begin{equation}
\lbeq{F12bds}
    \|F_1\|_{\Ucal} \le c_L L^{\alpha'(j-j_m)_+}\|F_1(B)\|_{T_0},
    \quad
    \quad
    \|F_2(B)\|_{T_0} \le c  \|F_2\|_{\Ucal},
\end{equation}
i.e., the bound remains the same for $F_1$ but loses a helpful factor $L^{-\alpha'(j-j_m)_+}$
for $F_2$.  The bound on $F_1$ then implies, as in
\cite[\eqref{alpha-e:Rnormcomp}]{Slad17},
that
\begin{equation}
\lbeq{Rnormcomp}
    \|R_+\|_\Ucal \le O(L^{\alpha'(j-j_m)_+})\|R_+(B)\|_{T_0},
\end{equation}
and the bound \refeq{Rmain-g} follows from this as in
\cite[Section~\ref{alpha-sec:Rbound}]{Slad17}.

The introduction of observables does lead to a change in the bounds on $F,W,P$
in \cite[Lemma~\ref{alpha-lem:W1}]{Slad17}, due to the weakened estimate for $F_2$
in \refeq{F12bds}.  The change is to replace the factors $L^{-(\alpha+\alpha')(j-j_m)_+}$
and $(c/L)^{-(\alpha+\alpha')(j-j_m)_+}$
in the three upper bounds of \cite[Lemma~\ref{alpha-lem:W1}]{Slad17} by the worse factor
$\chiL^2=L^{-\frac 12 \alpha(j-j_m)_+}$.  Since we seek an upper bound
which includes the factor $\chiL^2$ in $\epdV^2$, the weakened bounds
remain more than good enough.

For general reasons, $\pi_{ab}W=0$ \cite[Proposition~\ref{IE-prop:Wbounds}]{BS-rg-IE},
so there can be no such term in $W$.  Thus, in the proof of
\cite[Lemma~\ref{alpha-lem:W1}]{Slad17}, only one factor
$L^{-\alpha'(j-j_m)_+}$ can be lost by application of \refeq{F12bds},
not two.
Also, by direct calculation, the relevant contribution to $F$ is
$F_C(\lambda_a\sigma_a\varphi_a^1,\lambda_b\sigma_b\varphi_b^1)
= \half \lambda_a\lambda_bC_{a,b}\sigma_a\sigma_b$,
whose $T_0(\ell)$ norm
is given as in the first inequality of \refeq{Csigsig} to be at most $L^{-(\alpha-\alpha')(j-j_m)}$,
which is better than the required $\chiL^2$.
The bound on $P$ follows from the bounds on $F,W$
as in \cite[Proposition~4.1]{BS-rg-IE}.
\end{proof}

In the absence of observables, Theorem~\ref{thm:step-mr}
is used in \cite{Slad17} to construct a global RG flow $(g_j,\nu_j,K_j^\varnothing)$
that remains in the RG domain for all $j$.
This requires tuning the initial $\nu$ to a mass-dependent critical value
$\nu_0^c(m^2)$; this value
converges to the critical point $\nu_c(g;n)$ as $m^2 \downarrow 0$
(see \cite[\eqref{alpha-e:nuceta}--\eqref{alpha-e:nuc}]{Slad17}).
Throughout the remainder of the present paper,
we always take $(g_j,\nu_j)$ to be this global flow of coupling constants.  For general
reasons this flow is the same in the presence of observables as in their absence:
see \cite[\eqref{step-e:piVKplus}--\eqref{step-e:plusindep}]{BS-rg-step}.
The main task for the proof of Theorem~\ref{thm:2ptfcn} is to apply the
estimates of Theorem~\ref{thm:step-mr} to control, in addition,
the flow of the observable coupling constants $\lambda$ and $\delta q$,
and the observable part of the coordinate $K$.
The flow of $\delta q$ and $K$
is analysed
as in the 4-dimensional nearest-neighbour case \cite{BBS-saw4,ST-phi4}.

The flow of $\lambda $ is marginal, for the same reasons as in the 4-dimensional case.
In \cite{BBS-saw4,ST-phi4}, the perturbative approximation \refeq{lambdapt2} to the
recursion for $\lambda$ is solved along the lines of the rough computation
\begin{align}
\lambda_j
&=
\prod_{k=1}^{j-1} (1-\delta [\nu w^{(1)}])
=
\exp \left[ \sum_{k=1}^{j-1} \log ( 1-\delta [\nu w^{(1)}]) \right]
\nnb &
\approx \exp \left[ -\sum_{k=1}^{j-1} \delta [\nu w^{(1)}] \right]
=
\exp \left[ -\nu_{j} w_{j}^{(1)}] \right]
\approx 1 - \nu_{j}  w_{j}^{(1)}
.
\label{e:approxlambda}
\end{align}
In \cite{BBS-saw4,ST-phi4}, the errors introduced by the map $R_+$ into \refeq{lambdapt2}
were summable over all scales because of the decay of the marginal coupling constant $g_j$
with the scale (Gaussian fixed point), and the above computation survives the introduction
of these errors.

For the long-range model, the fixed point is non-Gaussian,
and the corrections due to $R_+$ are not summable.
Instead of trying to follow the route laid out in \cite{BBS-saw4,ST-phi4},
we derive an exact relation
between $\lambda_{a,j} $ and the known bulk coupling constants,
similar to \refeq{approxlambda},
which gives better control of its flow than the recursion.
This is done in Section~\ref{sec:lambda}.

\section{Flow of \texorpdfstring{$\lambda$}{lambda}}
\label{sec:lambda}

According to \refeq{lambdapt2} and Theorem~\ref{thm:step-mr}, the flow of $\lambda_{a,j} $
under the RG map is nontrivial only until scale $j_{ab}-1$,
and stops beyond this scale.
Conversely, $q_{a,j}=0$ for $j < j_{ab}$, and the flow of $q_{a,j}$ is
nontrivial only for scales $j \ge j_{ab}$.
Our goal now is to determine the form of the flow until scale $j_{ab}$.
Since we later take the limit $m^2\downarrow 0 $ before studying large
$j_{ab}\sim\log_L \vert a-b\vert $, we can and do assume that $j_{ab}< j_m$.
We will prove the following proposition.

\begin{prop}
\label{prop:lambda}
Let $n \ge 0$, let $L$ be sufficiently large, let $\Lambda_N$ be the torus of period $L^N$,
and let $\epsilon$ be sufficiently small.
Let $g \in [\frac{63}{64}\gLfix,\frac{65}{64}\gLfix]$, and let
$m^2 \in [L^{-\alpha(N-1)}, \delta]$ with $\delta >0$ sufficiently small.
Let $j_{ab}< j_m  < N$.
Let $g_0=g$, let $\nu_0$ be the critical value $\nu_0^c(m^2)$ constructed
in \cite{Slad17}, and let $\lambda_{a,0}=\lambda_{b,0}=1$.
Then the RG map can be iterated to scale $j_{ab}$, i.e., it produces
a sequence $(V_j,K_j)\in \domRG_j$ with initial condition $(V_0,0)$, such that \refeq{zzeta} holds
for all $j \le j_{ab}$ with $I_j=I_j(V_j)$ and $\zeta_j = \sum_{k=1}^j \delta\zeta_j$.
Moreover, $q_{x,j}=0$, and for the component $\lambda_{x,j} $ of this flow
we have the stronger statement
\begin{equation}\label{e:lambdaflow}
        \lambda_{x,j} = 1 -\nu_j w_j^{(1)} +O(\gLfix^2)
        \quad\quad
        (j < j_{ab}, \; x=a,b).
\end{equation}
\end{prop}

The proof of Proposition~\ref{prop:lambda} is given in Section~\ref{sec:intbyparts} below.
Its statement holds trivially at $j=0 $, and will be established inductively for higher scales.
The induction for the bulk quantities $U_j^\varnothing,K_j^\varnothing $ is the result of
\cite{Slad17}, and is unaffected by the presence of observables.

The main additional ingredient for the induction of the observable parts is to establish
the flow \refeq{lambdaflow} of $\lambda_{a,j}$. To achieve this, in Lemma~\ref{lem:ibp}
we use integration by parts to obtain a relation between $\lambda_{a,j} $, quantities
of the bulk flow, and the observable parts of the coordinate $K_j$.
This is achieved by taking suitable derivatives of the identity
$Z_j=\Ex_{w_j} \theta Z_0$.
The contribution due to $K_j$ is bounded uniformly
in the volume using a cluster expansion, in Section~\ref{sec:cluster}.

The formula \refeq{lambdaflow} for $\lambda_{x,j}$
has a natural counterpart for  the nearest-neighbour 4-dimensional case,
with error term $O(g_j^2)$ instead of $O(\gLfix^2)$.  In that context,
$-\nu_j w_j^{(1)} +O(g_j^2) \to 0$ as $j \to \infty$.  This
provides insight into
the fact that
$\lim_{j\to\infty} \lambda_{x,j}=1$
in \cite[Lemma~\ref{saw4-lem:lamlim}]{BBS-saw4} and \cite[Corollary~\ref{phi4-cor:vx}]{ST-phi4}.
For the long-range model considered in Proposition~\ref{prop:lambda},
the non-Gaussian fixed point leads to a limit which is not equal to $1$.

\subsection{Integration by parts}
\label{sec:intbyparts}

For notational convenience we restrict attention to $n \ge 1$; small modifications
apply for $n=0$.
Recall that $\bar D$
and $D_{\sigma_a}$ are defined above Corollary~\ref{cor:W0}, and that $Z_j=\Ex_{w_j} \theta Z_0$.
Let
\begin{equation}
\lbeq{zdef}
        z_j=z_j(\Lambda) = e^{-\zeta_j}\frac{Z_j(\Lambda)}{I_j(\Lambda)},
        \quad\quad
        \Lcal_j=\Lcal_j(\Lambda) = \log z_j(\Lambda).
\end{equation}
Then we have
\begin{equation}
\lbeq{ZLcal}
Z_j =
e^{\zeta_j}I_j(\Lambda) z_j(\Lambda)
=
e^{\zeta_j}I_j(\Lambda) e^{\Lcal_j(\Lambda)}.
\end{equation}
The existence of the logarithm $\Lcal_j $ is discussed
in Section~\ref{sec:cluster}, where it is constructed as an element of a
Banach space $T_0(\ell_j)$ which only examines derivatives at zero field,
using a cluster expansion.
Bounds on $\Lcal_j $ and its derivatives at zero
field are proved in
Proposition~\ref{prop:Lcalbds} below.

\begin{lemma}
\label{lem:ibp}
The functions $I_j$ and $\Lcal_j$ are related by the identity
\begin{equation}
\lbeq{ibpLcal}
        \bar D D_{\sigma_a} I_j(\Lambda) + \bar D D_{\sigma_a} \Lcal_j(\Lambda)
        =
        1
        +
        \frac{1}{|\Lambda|} w_j^{(1)}
        \left[ \bar D^2 I_j^\varnothing (\Lambda) + \bar D^2 \Lcal^\varnothing_j(\Lambda)
        \right].
\end{equation}
\end{lemma}

\begin{proof}
By definition, followed by Gaussian integration by parts,
\begin{align}
        D_{\sigma_a} \Ex_{w_j} \theta Z_0
        & =  \Ex_{w_j} (\varphi_a^1+ \zeta_a^1) Z_0^\varnothing(\varphi +\zeta)
        \nnb &
        = \varphi_a^1 \Ex_{w_j} Z_0^\varnothing(\varphi +\zeta)
        + \sum_{y \in \Lambda} w_{j;a,y}\Ex_{w_j} \frac{\partial}{\partial\zeta_y^1}
        Z_0^\varnothing(\varphi +\zeta).
\end{align}
On the right-hand side, $\frac{\partial}{\partial\zeta_y^1}$ can be replaced by
$\frac{\partial}{\partial\varphi_y^1}$, and the latter commutes with the expectation.
Then application of $\bar D = \sum_{x\in\Lambda}\frac{\partial}{\partial\varphi_x^1}|_{\varphi=0}$
gives
\begin{align}
        \bar D D_{\sigma_a} \Ex_{w_j} \theta Z_0
        & =
        \Ex_{w_j}  Z_0^\varnothing
        + \sum_{y \in \Lambda} w_{j;a,y}
        \sum_{x \in \Lambda}
        \frac{\partial^2}{\partial\varphi_y^1\partial\varphi_x^1}
        \Big|_{\varphi=0}
        \Ex_{w_j} Z_0^\varnothing(\varphi+\zeta),
\end{align}
which by translation invariance and by definition of $Z_j$ is the same as
\begin{align}
        \bar D D_{\sigma_a} Z_j & =
        Z_j^\varnothing|_{\varphi=0}
        +   w_j^{(1)} \frac{1}{|\Lambda|} \bar D^2
        Z_j^\varnothing.
\lbeq{ddzj}
\end{align}
Now we divide both sides of \refeq{ddzj} by $Z_j^\varnothing|_{\varphi=0}$
and use \refeq{ZLcal}. Since $I_j|_{\varphi=0}=1$, and since
$\bar D Z^\varnothing_j=\bar D I_j|_{\sigma_a=\sigma_b=0} = D_{\sigma_a} I_j\vert_{\varphi=0} = 0 $
by symmetry, the result is \refeq{ibpLcal}.
\end{proof}

Note that the right-hand side of \refeq{ibpLcal} involves only bulk quantities,
while the left-hand side depends
on $\lambda_{a,j} $ through $I_j(\Lambda) $ and $D_{\sigma_a} \Lcal_j(\Lambda) $,
and also on the observable part of the irrelevant coordinate $K_j $ (through $ D_{\sigma_a}\Lcal_j(\Lambda) $).
For the explicit terms, we have the following identities.

\begin{lemma}
\label{lem:Ibds}
For $j \le N$ and $V \in \Vcal$,
\begin{equation}
        \bar D^2 I_j^\varnothing(\Lambda)
        =
        - \nu_j |\Lambda|,
\end{equation}
and if $j < j_{ab}$ then
\begin{equation}
\bar D D_{\sigma_a} I_j(\Lambda) = \lambda_{a,j}.
\end{equation}
\end{lemma}

\begin{proof}
We differentiate the formula
$I_{j}(\Lambda) = e^{-V_j(\Lambda)}\prod_{B\in \Bcal_j(\Lambda)} (1+ W_j(V_j,B))$,
which is \refeq{Idef}.
We apply the product rule, Corollary~\ref{cor:W0}, and the facts that
$\bar D V_j^\varnothing =0$ and $I_j\vert_{\varphi=0}=1 $,
to obtain
\begin{align}
\lbeq{DDI}
\bar D^2 I_{j}^\varnothing (\Lambda)
&= -\nu_j\vert\Lambda\vert  + \sum_{B\in \Bcal _j(\Lambda)}\bar D^2 W_{j}^\varnothing(V_j,B)
 = -\nu_j \vert\Lambda\vert  .
\end{align}
Similarly, for $j < j_{ab}$, we also use
$D_{\sigma_a} V_j(\Lambda) = -\lambda_{a,j} \varphi^1_a$
to obtain
\begin{equation}
\lbeq{DDaI}
\bar D D_{\sigma_a} I_{j}(\Lambda)
    = \lambda_{a,j} + \sum_{B\in \Bcal_j(\Lambda)}\bar D  D_{\sigma_a} W_{j}(V_j,B)
    = \lambda_{a,j} ,
\end{equation}
and the proof is complete.
\end{proof}

We now state our bounds for the terms in \refeq{ibpLcal} involving $\Lcal_j$.
The hypothesis $(V_j,K_j)\in \domRG_j  $
of Proposition~\ref{prop:Lcalbds} will be established inductively.

\begin{prop}
\label{prop:Lcalbds}
Let $j \le j_{ab}$, let $\Lcal_j$ be defined as in \refeq{zdef}, and assume that
$Z_j = e^{\zeta_j}(I_j\circ K_j)$ with $I_j=I_j(V_j)$ and $(V_j,K_j) \in \domRG_j$.
Then there is a
constant $c_1>0$ such that
\begin{align}
\lbeq{DL}
|\bar D D_{\sigma_a} \Lcal_j(\Lambda)|  & \le c_1 \bar s^2,
\quad\quad
|\bar D^2 \Lcal_j^\varnothing (\Lambda) |    \le c_1 \vert\Lambda\vert L^{-\alpha j}\bar s^3.
\end{align}
\end{prop}

We defer the proof of Proposition~\ref{prop:Lcalbds} to Section~\ref{sec:cluster}.

\begin{proof}[Proof of Proposition~\ref{prop:lambda}]
The proof is by induction on $j$.
The statement of Proposition~\ref{prop:lambda} for $j=0$ is trivial.
Without loss of generality, we consider the case $x=a$.
We assume that we have \refeq{zzeta} for $Z_k$ with
$(V_k,K_k)$ constructed inductively using the RG map for $k \le j$,
and we make the constant in the hypothesis \refeq{lambdaflow} explicit by assuming that,
with $c_1$ from \refeq{DL},
\begin{equation}
\lbeq{lambdaflowM}
    |\lambda_{a,j}-1+\nu_jw_j^{(1)}| \le 2c_1 \gLfix^2.
\end{equation}
Then we have \refeq{zzeta}
with a pair of RG coordinates $(V_j,K_j)\in\domRG_j $, satisfying in addition
\refeq{lambdaflowM}.
Theorem~\ref{thm:step-mr} guarantees the existence of
RG coordinates
$(U_{j+1},K_{j+1})=(\delta \zeta_{j+1}=\delta u_{j+1},V_{j+1},K_{j+1})$
at scale $j+1$ such that $Z_{j+1}$ obeys \refeq{zzeta}, with
$U_{j+1}
= \PT_{j}(V_j) + R_{j+1}(V_j,K_j) $, and bounds on $R_{j+1}(V_j,K_j) $
and $K_{j+1} $ as in \refeq{Rmain-g}.

It has been proved in \cite{Slad17} that the bulk part of $V_{j+1} $ lies in
$\domRG_{j+1} $. The second bound in \refeq{Rmain-g} is sufficient to guarantee
that $K_{j+1} $ also lies in $\domRG_{j+1} $. Therefore, to complete the proof that
$(V_{j+1},K_{j+1})\in\domRG_{j+1} $, we only need to show that
$\vert \lambda_{j+1}\vert < C_{\Dcal} $, where $C_{\Dcal} >1 $ is the constant
in \refeq{DVdef-obs}. By \refeq{lambdapt2}, and by the first bound of \refeq{Rmain-g}
together with the definition of the norm in \refeq{Ucalnorm},
we have $ \lambda_{j+1} = (1+O(\bar s)) \lambda_j  + O(\bar s^2) $.
It now follows immediately from
\refeq{lambdaflowM} that
$0< \lambda_{j+1}  = 1+O(\bar s)<C_{\Dcal} $,
since $\bar s $ can be chosen small enough. This proves that
$(V_{j+1},K_{j+1})\in\domRG_{j+1} $.

To complete the induction, we must prove \refeq{lambdaflowM} with $j $ replaced by $j+1$.
Since \refeq{lambdaflow} is only required for scales below the coalescence scale,
we may assume
here that $j+1<j_{ab}$.
The bounds of
Proposition~\ref{prop:Lcalbds} at scale $j+1$ can be applied, since the hypothesis
$(V_{j+1},K_{j+1})\in\domRG_{j+1} $ has now been verified.
Also, the hypothesis $j+1<j_{ab}$ of Lemma~\ref{lem:Ibds} is satisfied.
We use Lemma~\ref{lem:Ibds}
in conjunction with \refeq{ibpLcal}, and apply
Proposition~\ref{prop:Lcalbds}.
This gives
\begin{equation}
    \vert \lambda_{a,j+1} -1 + \nu_{j+1}w_{j+1}^{(1)}\vert
    \leq
    c_1\gLfix^2 + c_1 w_{j+1}^{(1)}L^{-\alpha (j+1)} \gLfix^3
    \le
    2c_1 \gLfix^2,
\end{equation}
by \refeq{w1bd} and by taking $\gLfix$ sufficiently small.
This advances \refeq{lambdaflowM} to scale $j+1$, and completes the proof.
\end{proof}

\subsection{Cluster expansion}
\label{sec:cluster}

In this section, we use a cluster expansion to construct a formula for
$\Lcal_j = \log z_j $ and prove Proposition~\ref{prop:Lcalbds}.
Let $p(X) = K_j(X)/I_j(X)$.
By \refeq{zzeta}, \refeq{zdef}, and by definition of the circle product,
\begin{equation}
z_j =  I_j(\Lambda)^{-1} (I_j\circ K_j)(\Lambda)
    = \sum_{X\in \Pcal_j(\Lambda)} p(X),
\end{equation}
where the term in the sum with $X=\varnothing$ is interpreted as $1$.
In the sum, we decompose $X\in \Pcal_j$ into its connected components $X_1,\ldots,X_n \in
\Ccal_j$,
which may be labelled in $n!$ different ways.
For $X,X'\in \Ccal_j$, we set $g(X,X')=-1$ if $X$ and $X'$ touch,
and otherwise set $g(X,X')=0$.
Using the component factorisation property of $K_j$, we obtain
\begin{equation}
    z_j = \sum_{n= 0}^\infty \frac1{n!}
    \sum_{X_1,\ldots,X_n\in\Ccal_j}  p(X_1)\cdots p(X_n) \prod_{1 \le i<j \le n}(1+g(X_i,X_j)),
\end{equation}
where the $n=0$ term is again interpreted as $1$. This has the form of the partition function
of a polymer system, as defined, e.g., in \cite[(1)]{Uelt04}.
It is a standard result, e.g., \cite{Uelt04,FP07}, that $\log z_j$ can  be written as a cluster
expansion and accurately bounded, provided the polymer activities $p(X)$ obey suitable estimates.
In the following proof, we discuss this in detail and invoke a convergence criterion from \cite{Uelt04};
see also \cite{FV17,Salm99} for pedagogical introductions to the cluster expansion.
The verification of the criterion from \cite{Uelt04} is an almost immediate consequence
of the norm estimates in the definition of the domain $\domRG_j $.

Since we are interested only in the derivatives of $\Lcal_j $ at zero external and
observable fields, we do not construct $\Lcal_j $ as a function of these fields
(even though this would also be possible for suitably small fields),
but rather as a Taylor polynomial (jet) of order $p_{\Ncal} $ in the fields
around zero. In other words, we work on the quotient of $\Ncal $ by the
ideal of elements of $F\in \Ncal$ with $\|F\|_{T_0(\ell_j)} = 0$.
On this quotient, the $T_0(\ell_j)$ seminorm becomes a norm, and the
quotient becomes a finite-dimensional Banach algebra.
This is discussed in detail in \cite[Section~\ref{step-sec:Knorms}, Appendix~\ref{step-sec:Banach}]{BS-rg-step}.
We adopt the point of view in the following that we work in this normed space,
and write simply $\|\cdot\|$ for $\|\cdot \|_{T_0(\ell_j)}$.
Although the results of \cite{Uelt04} are stated for complex-valued $p(X)$,
the proofs hold verbatim for values in any Banach algebra.
The completeness of the Banach algebra is important for the existence
of $\Lcal_j = \log z_j$, which is defined in terms of an infinite sum.

The estimates we use, for $(V_j,K_j) \in \domRG_j$
and $X \in \Ccal_j$, are:
\begin{equation}
    \|1/I_j(X)\| \le 2^{|X|},
    \quad\quad
    \|K_j(X)\| \le M \gLfix^{3+ a(|X|-2^d)_+},
\end{equation}
where $a>0$ is small; here, $|X| = |X|_j$ denotes the number of $j$-blocks in $X$.
The bound on $I_j^{-X}$ is a small adaptation of \cite[Proposition~\ref{IE-prop:Istab}]{BS-rg-IE}
to our long-range setting, and the bound on $K_j$ follows from the definition of the $\Wcal_j$
norm and \refeq{domRG}.
Absorbing the factor $2^{|X|}$ by replacing $M$ by $M'>M$, replacing
$a$ by $a'\in (0,a)$, and using the fact that $\gLfix$ is sufficiently small,
we conclude that the polymer activity obeys the bound
\begin{equation}
\lbeq{pXbd}
    \|p(X)\| \le M' \gLfix ^{3+ a'(|X|-2^d)_+}.
\end{equation}
The following lemma uses this bound and will be employed to verify the hypothesis of \cite[Theorem~1]{Uelt04}.

\begin{lemma}
\label{lem:cluster-hyp}
If $B\in\Bcal_j$, then for $\gLfix$ sufficiently small (depending only on $d$)
\begin{equation}
\sum_{Y\in\Ccal_j} |g(B, Y)| \|p(Y)\| e^{|Y|}
    \le
O(\gLfix^3),
\end{equation}
where the constant depends on $d$.
\end{lemma}

\begin{proof}
The number of connected polymers $Y \in\Ccal_j$ that touch a block $Y$
and have size $|Y|=n$ is at most $A^n$ for some $d$-dependent constant $A$.
Thus,
\begin{equation}
\sum_{Y\in\Ccal_j} |g(B, Y)| \|p(Y)\| e^{|Y|}
    \le
M' \gLfix^3 \sum_{n=1}^{|\Lambda|} (Ae)^n \gLfix ^{a' (n-2^d)_+}.
\end{equation}
We split the sum on the right-hand side into sums with $n \le 2^d$ and $n > 2^d$.
The first of these is a constant
that depends only on $d$. Taking $\gLfix$ small so that $A e \gLfix^{a'} < 1$,
the second sum is bounded as
\begin{equation}
\gLfix^{-a' 2^d} \sum_{n=2^d+1}^{|\Lambda|} (A e \gLfix^{a'})^n
    = O(\gLfix^{a'})
\end{equation}
with a $d$-dependent constant, which suffices.
\end{proof}

\begin{proof}[Proof of Proposition~\ref{prop:Lcalbds}]
By Lemma~\ref{lem:cluster-hyp}, if $\gLfix$ is sufficiently small, then
\begin{equation}
\sum_{Y\in\Ccal_j} |g(X, Y)| \|p(Y)\| e^{|Y|}
    \le \sum_{B\in\Bcal_j(X)}
    \sum_{Y \in \Ccal_j}
    |g(B, Y)| \|p(Y)\| e^{|Y|}
    \le |X|,
\end{equation}
which verifies the hypothesis \cite[(3)]{Uelt04} with $a(A)=|A|$.
Also by Lemma~\ref{lem:cluster-hyp},
\begin{equation}
\label{e:pexp-bd}
\sum_{Y \in \Ccal_j}\|p(Y)\|e^{|Y|}
    \le \sum_{B\in\Bcal_j} \sum_{Y \in \Ccal_j} |g(B, Y)| \|p(Y)\| e^{|Y|}
    \le O(
    L^{(N-j)d}
    \gLfix^3) < \infty,
\end{equation}
which verifies the other hypothesis of \cite[Theorem~1]{Uelt04}.

Let $u(X_1,\ldots,X_n)$ denote the Ursell function, defined in \cite[(2)]{Uelt04}
(with $g$ written as $\zeta$).
We conclude from \cite[Theorem~1]{Uelt04} that $\Lcal_j$ is given by the
absolutely convergent sum
\begin{equation}
\label{e:cexpL}
\Lcal_j =
    \sum_{n=1}^\infty
    \sum_{X_1,\ldots,X_n\in\Ccal_j(\Lambda)}
    p(X_1)\cdots p(X_n)u(X_1,\ldots,X_n)  ,
\end{equation}
and that for all $X_1\in \Ccal_j$ we have
\begin{equation}
\lbeq{cebd}
    \sum_{n=1}^\infty n
    \sum_{X_2,\ldots,X_n\in \Ccal_j}
    \|p(X_2)\| \cdots \|p(X_n)\|
    |u(X_1,\ldots, X_n)|
        \le
    e^{|X_1|},
\end{equation}
with the $n=1$ term in \refeq{cebd} interpreted as $1$.

By \refeq{cexpL}--\refeq{cebd} (the factor $n$ in \refeq{cebd} is not needed)
and \eqref{e:pexp-bd},
\begin{equation}
\lbeq{Lcalbd}
    \Vert \Lcal_j^\varnothing\Vert
    \le
    \sum_{X_1\in\Ccal_j} \|p(X_1)\| e^{|X_1|} = O(L^{(N-j)d} \gLfix^3).
\end{equation}
Similarly, for $\Vert D_{\sigma_a} \Lcal_j \Vert$,
we use the product rule for differentiation, this time using the factor
$n$ (due to the product rule) in \refeq{cebd}.  With the definition of the
$T_0$ seminorm in \refeq{Tphiobs} and of $\ell_{\sigma,j}$ in \refeq{weightsigma}
for $j \le j_{ab}$, we obtain
\begin{equation}
\lbeq{Lcalabd}
    \Vert D_{\sigma_a} \Lcal_j \Vert
    \le
    \sum_{X\in\Ccal_j : X \ni a} \|D_{\sigma_a} p(X)\| e^{|X|}
    =
    \sum_{X\in\Ccal_j : X \ni a} \ell_{\sigma,j}^{-1} \|  p(X)\| e^{|X|}
    = O(\ell_j \gLfix^2).
\end{equation}
For $F \in \Ncal$, $\bar D F = \pair{F,\1^1}_0$
and $\bar D^2 F =\pair{F,\1^{2}}_0$, with the test functions $\1^1,\1^2$ of
Lemma~\ref{lem:Loc}.
These two test functions have $\Phi_j$-norms
(as defined, e.g., in \cite[\eqref{alpha-e:Phinorm}]{Slad17})
\begin{equation}
    \|\1^1\|_{\Phi_j} = \ell_j^{-1}, \quad\quad  \|\1^{2}\|_{\Phi_j} = \ell_j^{-2}.
\end{equation}
Therefore, for $m=1,2$,
\begin{equation}
    |\bar D^m F| \le \|F \Vert
    \ell_j^{-m}
    .
\end{equation}
In particular, since $L^{Nd}=|\Lambda|$,
\begin{align}
    |\bar D D_{\sigma_a} \Lcal_j|
    &
    \le \Vert D_{\sigma_a} \Lcal_j \Vert
    \ell_j^{-1} = O(\gLfix^2),
    \\
    |\bar D^2 \Lcal_j^\varnothing|
    &
    \le
    \Vert \Lcal_j^\varnothing \Vert
    \ell_j^{-2}
    = O(|\Lambda| L^{-\alpha j} \gLfix^3),
\end{align}
and the proof is complete.
\end{proof}

\section{Full RG flow and proof of Theorem~\ref{thm:2ptfcn}}
\label{sec:qflow}

In Proposition~\ref{prop:lambda}, the
RG flow $(\zeta_j,V_j,K_j) $ is constructed for scales $j\leq j_{ab} $.
The sequence $\zeta_j$ of \refeq{zetadef} contains in particular the coupling constants
$q_{a,j},q_{b,j}$; recall that $q_{x,j}=0$ for $j \le j_{ab}$.
In Section~\ref{sec:q}, we apply Theorem~\ref{thm:step-mr} inductively to continue the
RG flow $(\zeta_j,V_j,K_j) $ to scales $j_{ab} < j \le N $.
Using the extended flow, we prove Theorem~\ref{thm:2ptfcn} in Section~\ref{sec:2pt}.
The analysis proceeds as in \cite{BBS-saw4,ST-phi4}.

Once the RG flow has been extended to all scales, the combination of \refeq{ZN} and
\refeq{zzeta} gives, at the final scale $j=N$, the representation
\begin{equation}\label{e:zN}
    \Ex_C  e^{- V_{0} (\Lambda)} = Z_N \big|_{\varphi=0}
     = e^{\zeta_N}(I_N(\Lambda) + K_N(\Lambda) )\big|_{\varphi=0}.
\end{equation}
From this, we apply \refeq{corrdiff} to to calculate the two-point function as
\begin{align}
\label{e:DaDbPN}
    G_{a,b,N}(g,\nu;n)
    &=
    D_{\sigma_a\sigma_b}^2 \log \Ex_C e^{- V_{0} (\Lambda)}
    = \frac{1}{2}(q_{a,N} + q_{b,N}) + A_N,
\end{align}
with
\begin{align}
    A_N &=
    \frac{D^2_{\sigma_a\sigma_b}K_{N}}{1 +  K_{N}^\varnothing}\Big|_{\varphi=0}
    - \frac{\left(D_{\sigma_a}K_{N}\right)
    \left(D_{\sigma_b}K_{N}\right)}{(1 +  K_{N}^\varnothing)^2}\Big|_{\varphi=0}
    .
\end{align}

\subsection{Flow of \texorpdfstring{$q$}{q}}
\label{sec:q}

The next proposition states that the RG flow exists for scales $j_{ab}\leq j\leq N $, and in particular analyses the flow of $q$ and establishes control on the terms of the right-hand side
of \eqref{e:DaDbPN}, which is needed to prove Theorem~\ref{thm:2ptfcn}.

\begin{prop}
\label{prop:qN}
Let $n \ge 0$, let $L$ be sufficiently large,
let $\Lambda_N$ be the torus of period $L^N$,
and let $\epsilon$ be sufficiently small.
Let $g \in [\frac{63}{64}\gLfix,\frac{65}{64}\gLfix]$, and let
$m^2 \in [L^{-\alpha(N-1)}, \delta]$ with $\delta >0$ sufficiently small.
Suppose that $j_{ab}<j_m$.
Starting with $(V_{j_{ab}},K_{j_{ab}})$ produced by Proposition~\ref{prop:lambda},
the RG map can be iterated to scale $N$, i.e., it produces
a sequence $(V_j,K_j)\in \domRG_j$ such that \refeq{zzeta} holds
for all $j \le N$ with $I_j=I_j(V_j)$ and $\zeta_j = \sum_{k=1}^j \delta\zeta_j$.
The $q_{x,j}$ component of $\zeta_j$ is given by
\begin{equation}\label{e:q}
q_{x,j} = \lambda_{a, j_{ab}} \lambda_{b, j_{ab}}  w_{j; a,b} + \sum_{i = j_{ab}}^{j - 1} r_{x,i}
    \quad\quad
(x = a, b)
\end{equation}
with
\begin{equation}
\lbeq{rqbd}
|r_{x,i}| \le
O(\gLfix) \frac{1}{|a-b|^{d-\alpha}}
4^{-(i-j_{ab})_+}.
\end{equation}
Moreover,
\begin{align}
\lbeq{ANlim}
\lim_{N \to \infty} A_N=0.
\end{align}
\end{prop}

\begin{proof}
For $j=j_{ab} $, we have $(V_{j_{ab}},K_{j_{ab}})\in \domRG_{j_{ab}}$
by Proposition~\ref{prop:lambda}.
Also, \refeq{q}--\refeq{rqbd} hold trivially, since $r_{x,j_{ab}}=0$
by Theorem~\ref{thm:step-mr} and hence
$q_{x,j_{ab}}=\lambda_{a, j_{ab}} \lambda_{b, j_{ab}}  w_{j; a,b}$ by \refeq{qpt2}.

We fix $j \ge j_{ab}$ and assume inductively that
\refeq{zzeta} holds with a pair of RG coordinates $(V_j,K_j)\in \domRG_j  $
and that \refeq{q}--\refeq{rqbd} hold.
 As in the proof of
Proposition~\ref{prop:lambda}, Theorem~\ref{thm:step-mr} guarantees the existence
of RG coordinates $(V_{j+1},K_{j+1}) $ at scale $j+1$, with
$V_{j+1} = \PT_{j+1}(V_j) + R_{j+1}(V_j,K_j) $, and bounds on
$R_{j+1}(V_j,K_j) $ and $K_{j+1} $ as in \refeq{Rmain-g}.

As before, it has been proved in \cite{Slad17} that the bulk part of $V_{j+1} $ lies in
$\domRG_{j+1} $. The coordinate $\lambda_{a,j}=\lambda_{a,j_{ab}} $ remains constant for
$j> j_{ab} $, and thus still lies in $\domRG_{j+1} $. As before, the second bound
in \refeq{Rmain-g} is sufficient to guarantee that $K_{j+1} $ also lies in $\domRG_{j+1} $.
This shows that $(V_{j+1},K_{j+1})\in\domRG_{j+1} $.

We now show that $q_{a,j+1} $ satisfies \refeq{q} at scale $j+1$ and that \refeq{rqbd} holds.
Using \refeq{q} and \refeq{qpt2}, and denoting by $r_{a,j} $ the
component of $R_{j+1}(U_j,K_j) $ corresponding to the component $q_a$, we see that
\begin{equation*}
q_{a,j+1} = q_{a,j} + \lambda_{a,j_{ab}}\lambda_{b,j_{ab}} C_{j+1;a,b} + r_{a,j}
=  \lambda_{a, j_{ab}} \lambda_{b, j_{ab}}  w_{j+1; a,b} + \sum_{i = j_{ab}}^{j} r_{x,i},
\end{equation*}
verifying \refeq{q} at scale $j+1$.
By definition of the norm in \refeq{Ucalnorm}
and by our assumption that $j_{ab}<j_m$, Theorem~\ref{thm:step-mr} gives the bound
\begin{align}
\lbeq{goodrq}
r_{x,j} \le
\ell_{\sigma,j+1}^{-2} \|R_{j+1}\|
\le
1_{j \ge j_{ab}}
\ell_{\sigma,j+1}^{-2}
O(\gLfix^3)
& =
1_{j \ge j_{ab}}
L^{- j_{ab}(d-\alpha)}
4^{-(j-j_{ab})_+}O(\gLfix )
,
\end{align}
which proves \refeq{rqbd} since $L^{-j_{ab}(d-\alpha)}=O(|a-b|^{-(d-\alpha)})$
by \refeq{jabbds}.

Finally, we write $D^k_\sigma$ to mean
no derivative for $k=0$, the derivative with respect to $\sigma_a$ or $\sigma_b$
for $k=1$, and the second derivative with respect to $\sigma_a,\sigma_b$ for $k=2$.
Since $(V_N,K_N)\in\domRG_N$, it follows
from \refeq{Rmain-g}, with the fact that the $\Wcal_N$
norm bounds the $T_0(\ell_N)$ norm, that
\begin{align}
\label{e:Kg1}
    |D_{\sigma}^l K_{N}(\Lambda)|_{\varphi=0}|
    &
    \le
    \ell_{\sigma,N}^{-l}\CRG \chicCov_N^3 \gLfix^3
    \le
    O(\gLfix^{3-l})
    \chicCov_N^3
    \left( 2^{-(N-j_{ab})_+}
    L^{-\frac 12 j_{ab}(d-\alpha)}
    \right)^l
    .
\end{align}
Since $\chicCov_N \to 0$ as $N \to \infty$, this implies \refeq{ANlim} and completes the proof.
\end{proof}

\subsection{Proof of Theorem~\ref{thm:2ptfcn}}
\label{sec:2pt}

With \refeq{DaDbPN} and Proposition~\ref{prop:qN}, it is now straightforward to complete
the proof of our main result. In the proof, we write
$C_{a,b}$ for the massless free two-point function
$((-\Delta)^{\alpha/2})^{-1}_{a,b}$ on $\Zd$.  According to \refeq{resolventasy},
$C_{a,b} \asymp |a-b|^{-(d-\alpha)}$. The proof uses the following lemma.

\begin{lemma}
\label{lem:mcts}
Let $\nu_0 = \nu_0^c(m^2)$. Then for any $j < N$ the map $m^2 \mapsto (V_j, K_j)$
is continuous for $m^2 \in [0, L^{-\alpha j}]$. Moreover, the sequence $V_j$ is
independent of $N$. In particular, for any $j < \infty$, the maps $V_j, K_j, R_{+,j}$
depend continuously on $m^2$ at $m^2 = 0$.
\end{lemma}

\begin{proof}
We show by induction that $(V_j, K_j)$ depends continuously on
$m^2 \in [0, L^{-\alpha j}]$. The case $j = 0$ follows from
\cite[\eqref{alpha-e:nu0}]{Slad17} and \cite[Corollary~\ref{alpha-cor:mu0}]{Slad17}.
Now suppose the inductive hypothesis holds for some $j \ge 0$. Then the case $j + 1$
follows from \eqref{e:lambdapt2}, \cite[Lemma~\ref{alpha-lem:wlims}]{Slad17},
and Theorem~\ref{thm:step-mr}. The fact that $V_j$ is independent of $N$ is
\cite[Proposition~\ref{step-prop:VZd}]{BS-rg-step}.
\end{proof}

\begin{proof}[Proof of Theorem~\ref{thm:2ptfcn}]
We first take the limit $N \to \infty$,
then take the limit $m^2 \downarrow 0$,
and finally consider large $|a-b|$.
By \refeq{q} with $j=N$,
\begin{equation}
    q_{x,N}
    =
    \lambda_{a, j_{ab}} \lambda_{b, j_{ab}}  w_{N; a,b} + \sum_{i = j_{ab}}^{N - 1} r_{x,i}
    .
\end{equation}
By Proposition~\ref{prop:qN},
the remainder term is bounded uniformly in $N$ and
in $m^2 \in [L^{-\alpha (N-1)},\delta]$ by
\begin{equation}
    \Big| \sum_{i = j_{ab}}^{N - 1} r_{x,i} \Big|
    \le
    O(\gLfix) \frac{1}{|a-b|^{d-\alpha}}
    \le
    O(\gLfix) C_{a,b}.
\end{equation}
By dominated convergence, and by the continuity of $r_{x,i}$
(a component of $R_{i+1}$) at $m^2=0$ guaranteed by Lemma~\ref{lem:mcts},
$\lim_{m^2 \downarrow 0}\lim_{N \to \infty}\sum_{i = j_{ab}}^{N - 1} r_{x,i}$ exists
and is bounded by $O(\gLfix) C_{a,b}$.
For the main term, since $\lambda_{j_{ab}}=1+O(\gLfix)$ by Proposition~\ref{prop:lambda},
it follows from the definition of $w_N$ in \refeq{wjdef} (together with the fact
that that the covariance appearing in \refeq{qpt2} is always the infinite-volume one), that
\begin{align}
    \lim_{m^2 \downarrow 0}\lim_{N\to\infty}
    \lambda_{j_{ab}}^2  w_{N,a,b} = (1+O(\gLfix)) C_{a,b}
    .
\end{align}
The existence of the above limit as $m^2 \downarrow 0$ is a consequence of the
fact that $w_{\infty,a,b} = ((-\Lambda)^{\alpha/2}+m^2)^{-1}_{ab} \to C_{a,b}$,
together with the mass continuity of $\lambda_{j_{ab}}$, which follows from
Lemma~\ref{lem:mcts}.
We apply \refeq{ANlim} in \refeq{DaDbPN}, and find that
the critical two-point function obeys
\begin{equation}
        G_{a,b} =
        \lim_{N \to \infty} G_{a,b,N} =
        \frac{1}{2}(q_{a,\infty} + q_{b,\infty})
        =
        (1 +O(\gLfix)) C_{a,b}
        +O(\gLfix) C_{a,b} = (1 +O(\gLfix)) C_{a,b} .
\end{equation}
This completes the proof.
\end{proof}


\section*{Acknowledgements}
This work was supported in part by NSERC of Canada.
We thank Slava Rychkov for helpful correspondence, and an anonymous referee
for useful suggestions.


\end{document}

%% file: macros.tex
\usepackage{amsfonts}
\usepackage{amsmath,amssymb,amsthm}
\usepackage{appendix}
\usepackage{bbm} 
\usepackage{amsbsy}
\usepackage{enumerate}
\usepackage{cite}
\usepackage{enumerate}

\InputIfFileExists{./macros_local.tex}{}{}

\ifdefined\macrosPa
  \usepackage[textwidth=465pt,textheight=650pt,centering]{geometry} 
\else\ifdefined\macrosPb
  \usepackage[textwidth=500pt,textheight=650pt,centering]{geometry} 
\fi\fi

\ifdefined\macrosS


  \usepackage{mathptmx}
  \DeclareMathAlphabet{\mathcal}{OMS}{cmsy}{m}{n}
\fi

\ifdefined\macrosBirk
\else
\usepackage[dvips]{graphicx}
\fi

\ifdefined\macrosSB
\input ../def99c 
\else
\input def99c 
\fi

\UseSection   
\usepackage[usenames]{color}

\definecolor{bw}{RGB}{240, 120, 0}
\definecolor{at}{rgb}{0.0, 0.5, 0.0} 

\newcommand{\LTbar}{{\rm loc}}

\newcommand{\LT}{{\rm Loc}  }

\newcommand{\DV}{\Dcal}
\newcommand{\DVa}{\alpha}

\renewcommand{\to} {\rightarrow}

\newcommand{\R}{\Rbold}
\newcommand{\Z}{\Zbold}

\newcommand{\N}{\Nbold}

\newcommand{\1}{\mathbbm{1}}

\newcommand{\psib}{\bar\psi}

\newcommand{\Ex}{\mathbb{E}}

\newcommand{\chicCov}{{\chi}}

\newcommand{\pair}[1]{\langle #1 \rangle}

\newcommand{\cgam}{\gamma}

\newcommand{\pt}{{\rm pt}}

\newcommand{\lambdapt}{\lambda_{\mathrm{pt}}}

\newcommand{\qpt}{q_{\mathrm{pt}}}

\newcommand{\h}{\mathfrak{h}}
\ifdefined\macrosSB \else

\fi

\newcommand{\Iint}{\mathbb{I}}

\newcommand{\domRG}{\mathbb{D}}

\newcommand{\pp}{a}
\newcommand{\qq}{b}

\newcommand{\half}{\textstyle{\frac 12}}

\newcommand{\ddp}[2]{\frac{\partial #1}{\partial #2}}

\newcommand{\epdV}{\bar{\epsilon}}

\newcommand{\phib}{\bar\phi}

\newcommand{\Kspace}{\Kcal}
\newcommand{\CKspace}{\Ccal\Kspace}

\DeclareMathOperator{\Loc}{Loc} 

\ifdefined\macrosH
  \usepackage{xr-hyper}
  \usepackage{hyperref}
  \hypersetup{hypertexnames=false}
  \hypersetup{colorlinks,citecolor=blue,linkcolor=blue}  

\else\ifdefined\macrosHarxiv
  \usepackage{xr-hyper}
  \usepackage{hyperref}
  \hypersetup{hypertexnames=false, hidelinks}

\else
  \newcommand{\texorpdfstring}[2]{#1}
  \usepackage{xr}
\fi\fi

%% file: def99c.tex
\def\UseSection{
        \numberwithin{equation}{section}
	\theoremstyle{plain}
        \newtheorem{theorem}    {Theorem}[section]
        \DefineTheorems 
}

\def\DefineTheorems{
	
	\newtheorem{lemma}      [theorem] {Lemma}
	
	\newtheorem{prop}       [theorem] {Proposition}
	
	\newtheorem{cor}        [theorem] {Corollary}

	\theoremstyle{definition}
	\newtheorem{defn}       [theorem] {Definition}

	\theoremstyle{definition}

}

\newcommand{\bt}   {\begin{theorem}}
\newcommand{\et}   {\end  {theorem}}
\newcommand{\bl}   {\begin{lemma}}
\newcommand{\el}   {\end  {lemma}}
\newcommand{\bp}   {\begin{prop}}
\newcommand{\ep}   {\end  {prop}}
\newcommand{\bc}   {\begin{cor}}
\newcommand{\ec}   {\end  {cor}}
\newcommand{\bd}   {\begin{defn}}
\newcommand{\ed}   {\end  {defn}}

\newcommand{\ba}   {\begin{array}}
\newcommand{\ea}   {\end  {array}}
\newcommand{\be}   {\begin{enumerate}}
\newcommand{\ee}   {\end  {enumerate}}
\newcommand{\bi}   {\begin{itemize}}
\newcommand{\ei}   {\end  {itemize}}

\def\eq#1\en{\begin{equation}#1\end{equation}}  
\def\eqsplit#1\ensplit{
	\begin{equation}\begin{split}#1\end{split}\end{equation}
	}
\def\eqalign#1\enalign{
	\begin{align}#1\end{align}
	}
\def\eqmul#1\enmul{
	\begin{multline}#1\end{multline}
	}
\newcommand{\eqarrstar} {\begin{eqnarray*}} 
\newcommand{\enarrstar} {\end{eqnarray*}} 
\newcommand{\eqarray}   {\begin{eqnarray}} 
\newcommand{\enarray}   {\end{eqnarray}} 
\newcommand{\nnb}	{\nonumber \\} 

\newcommand{\lbeq}[1]  {\label{e:#1}}
\newcommand{\refeq}[1] {\eqref{e:#1}}    

\makeatletter
\newcommand{\labelcounter}[2]{{%
	\stepcounter{#1}
	\protected@write\@auxout{}%
	{\string\newlabel{#2}{{\csname the#1\endcsname}{\thepage}}}%
	{\ref{#2}}
	}}
\makeatother
\newcommand{\Nbold} {{\mathbb N}}

\newcommand{\Rbold} {{\mathbb R}}

\newcommand{\Zbold} {{\mathbb Z}}

\newcommand{\Bcal}   {\mathcal{B}} 
\newcommand{\Ccal}   {\mathcal{C}} 
\newcommand{\Dcal}   {\mathcal{D}}

\newcommand{\Ical}   {\mathcal{I}} 
 
\newcommand{\Kcal}   {\mathcal{K}} 
\newcommand{\Lcal}   {\mathcal{L}} 
 
\newcommand{\Ncal}   {\mathcal{N}} 
 
\newcommand{\Pcal}   {\mathcal{P}}

\newcommand{\Scal}   {\mathcal{S}} 
 
\newcommand{\Ucal}   {\mathcal{U}} 
\newcommand{\Vcal}   {\mathcal{V}} 
\newcommand{\Wcal}   {\mathcal{W}}

\newcommand{\Zd}    {{ {\Zbold}^d }}

\newcommand{\spose}[1] {{\hbox to 0pt{#1\hss}} }
\newcommand{\ltapprox} {\mathrel{\spose{\lower 3pt\hbox{$\mathchar"218$}}
 \raise 2.0pt\hbox{$\mathchar"13C$}}}
\newcommand{\gtapprox} {\mathrel{\spose{\lower 3pt\hbox{$\mathchar"218$}}
 \raise 2.0pt\hbox{$\mathchar"13E$}}}






%% file: alpha2pt.bbl
\begin{thebibliography}{10}

\bibitem{Abde07}
A.~Abdesselam.
\newblock A complete renormalization group trajectory between two fixed points.
\newblock {\em Commun. Math. Phys.}, {\bf 276}:727--772, (2007).

\bibitem{ACG13}
A.~Abdesselam, A.~Chandra, and G.~Guadagni.
\newblock Rigorous quantum field theory functional integrals over the $p$-adics
  {I}: {Anomalous} dimensions.
\newblock (2013).  \href{https://arxiv.org/abs/1302.5971}{arXiv:1302.5971}.

\bibitem{AF86}
M.~Aizenman and R.~Fern\'{a}ndez.
\newblock On the critical behavior of the magnetization in high dimensional
  {Ising} models.
\newblock {\em J. Stat. Phys.}, {\bf 44}:393--454, (1986).

\bibitem{BBS-phi4-log}
R.~Bauerschmidt, D.C. Brydges, and G.~Slade.
\newblock Scaling limits and critical behaviour of the $4$-dimensional
  $n$-component $|\varphi|^4$ spin model.
\newblock {\em J. Stat. Phys}, {\bf 157}:692--742, (2014).

\bibitem{BBS-saw4}
R.~Bauerschmidt, D.C. Brydges, and G.~Slade.
\newblock Critical two-point function of the 4-dimensional weakly self-avoiding
  walk.
\newblock {\em Commun.\ Math.\ Phys.}, {\bf 338}:169--193, (2015).

\bibitem{BBS-saw4-log}
R.~Bauerschmidt, D.C. Brydges, and G.~Slade.
\newblock Logarithmic correction for the susceptibility of the 4-dimensional
  weakly self-avoiding walk: a renormalisation group analysis.
\newblock {\em Commun.\ Math.\ Phys.}, {\bf 337}:817--877, (2015).

\bibitem{BBS-rg-pt}
R.~Bauerschmidt, D.C. Brydges, and G.~Slade.
\newblock A renormalisation group method. {III}. {Perturbative} analysis.
\newblock {\em J. Stat. Phys}, {\bf 159}:492--529, (2015).

\bibitem{BRRZ17}
C.~Behan, L.~Rastelli, S.~Rychkov, and B.~Zan.
\newblock A scaling theory for long-range to short-range crossover and an
  infrared duality.
\newblock {\em J. Phys. A: Math. Theor.}, {\bf 50}:354002, (2017).

\bibitem{BC15}
A.~Bendikov and W.~Cygan.
\newblock $\alpha$-stable random walk has massive thorns.
\newblock {\em Colloquium Mathematicum}, {\bf 138}:105--129, (2015).

\bibitem{BCT15}
A.~Bendikov, W.~Cygan, and B.~Trojan.
\newblock Limit theorems for random walks.
\newblock \emph{Stoch. Proc. Appl.}, in press.
\newblock (2015).
\href{https://arxiv.org/abs/1504.01759}{arXiv:1504.01759}

\bibitem{BPR-T14}
E.~Brezin, G.~Parisi, and F.~Ricci-Tersenghi.
\newblock The crossover region between long-range and short-range interactions
  for the critical exponents.
\newblock {\em J. Stat. Phys.}, {\bf 157}:855--868, (2014).

\bibitem{BMS03}
D.C. Brydges, P.K. Mitter, and B.~Scoppola.
\newblock Critical $({\Phi}^4)_{3,\epsilon}$.
\newblock {\em Commun. Math. Phys.}, {\bf 240}:281--327, (2003).

\bibitem{BS-rg-norm}
D.C. Brydges and G.~Slade.
\newblock A renormalisation group method. {I}. {Gaussian} integration and
  normed algebras.
\newblock {\em J. Stat. Phys}, {\bf 159}:421--460, (2015).

\bibitem{BS-rg-loc}
D.C. Brydges and G.~Slade.
\newblock A renormalisation group method. {II}. {Approximation by local
  polynomials}.
\newblock {\em J. Stat. Phys}, {\bf 159}:461--491, (2015).

\bibitem{BS-rg-IE}
D.C. Brydges and G.~Slade.
\newblock A renormalisation group method. {IV}. {Stability} analysis.
\newblock {\em J. Stat. Phys}, {\bf 159}:530--588, (2015).

\bibitem{BS-rg-step}
D.C. Brydges and G.~Slade.
\newblock A renormalisation group method. {V}. {A} single renormalisation group
  step.
\newblock {\em J. Stat. Phys}, {\bf 159}:589--667, (2015).

\bibitem{CS15}
L.-C. Chen and A.~Sakai.
\newblock Critical two-point functions for long-range statistical-mechanical
  models in high dimensions.
\newblock {\em Ann.\ Probab.}, {\bf 43}:639--681, (2015).

\bibitem{EPPRSV14}
S.~El-Showk, M.F. Paulos, D.~Poland, S.~Rychkov, D.~Simmons-Duffin, and
  A.~Vichi.
\newblock Solving the 3d {Ising} model with the conformal bootstrap {II.}
  $c$-minimization and precise critical exponents.
\newblock {\em J. Stat. Phys.}, {\bf 157}:869--914, (2014).

\bibitem{FP07}
R.~Fern\'{a}ndez and A.~Procacci.
\newblock Cluster expansion for abstract polymer models. {New} bounds from and
  old approach.
\newblock {\em Commun. Math. Phys.}, {\bf 274}:123--140, (2007).

\bibitem{FMN72}
M.E. Fisher, S.~Ma, and B.G. Nickel.
\newblock Critical exponents for long-range interactions.
\newblock {\em Phys. Rev. Lett.}, {\bf 29}:917--920, (1972).

\bibitem{FV17}
S.~Friedli and Y.~Velenik.
\newblock {\em Statistical Mechanics of Lattice Systems: A Concrete
  Mathematical Introduction}.
\newblock Cambridge University Press, Cambridge, in press.

\bibitem{Heyd11}
M.~Heydenreich.
\newblock Long-range self-avoiding walk converges to alpha-stable processes.
\newblock {\em Ann.\ I.\ Henri Poincar\'{e} Probab.\ Statist.}, {\bf
  47}:20--42, (2011).

\bibitem{HHS08}
M.~Heydenreich, R.~van~der Hofstad, and A.~Sakai.
\newblock Mean-field behavior for long- and finite range {Ising} model,
  percolation and self-avoiding walk.
\newblock {\em J. Stat. Phys.}, {\bf 132}:1001--1049, (2008).

\bibitem{Mitt13}
P.~Mitter.
\newblock Long range ferromagnets: renormalization group analysis.
\newblock \url{https://hal.archives-ouvertes.fr/cel-01239463}, (2013).

\bibitem{Mitt16}
P.K. Mitter.
\newblock On a finite range decomposition of the resolvent of a fractional
  power of the {Laplacian}.
\newblock {\em J. Stat. Phys.}, {\bf 163}:1235--1246, (2016).
\newblock Erratum: \emph{J. Stat. Phys.} {\bf 166}:453--455, (2017).

\bibitem{Mitt17}
P.K. Mitter.
\newblock On a finite range decomposition of the resolvent of a fractional
  power of the {Laplacian} {II}. {The} torus.
\newblock {\em J. Stat. Phys.}, {\bf 168}:986--999, (2017).

\bibitem{MS08}
P.K. Mitter and B.~Scoppola.
\newblock The global renormalization group trajectory in a critical
  supersymmetric field theory on the lattice ${{\mathbb Z}}^3$.
\newblock {\em J. Stat. Phys.}, {\bf 133}:921--1011, (2008).

\bibitem{PRRZ16}
M.F. Paulos, S.~Rychkov, B.C.~van Rees, and B.~Zan.
\newblock Conformal invariance in the long-range {Ising} model.
\newblock {\em Nucl. Phys. B}, {\bf 902}:246--291, (2016).

\bibitem{Sak73}
J.~Sak.
\newblock Recursion relations and fixed points for ferromagnets with long-range
  interactions.
\newblock {\em Phys. Rev. B}, {\bf 8}:281--285, (1973).

\bibitem{Salm99}
M.~Salmhofer.
\newblock {\em Renormalization: An Introduction}.
\newblock Springer, Berlin, (1999).

\bibitem{Slad17}
G.~Slade.
\newblock Critical exponents for long-range {$O(n)$} models below the upper
  critical dimension.
\newblock (2016).
\href{https://arxiv.org/abs/1611.06169}{arXiv:1611.06169}

\bibitem{ST-phi4}
G.~Slade and A.~Tomberg.
\newblock Critical correlation functions for the $4$-dimensional weakly
  self-avoiding walk and $n$-component $|\varphi|^4$ model.
\newblock {\em Commun. Math. Phys.}, {\bf 342}:675--737, (2016).

\bibitem{SYI72}
M.~Suzuki, Y.~Yamazaki, and G.~Igarashi.
\newblock Wilson-type expansions of critical exponents for long-range
  interactions.
\newblock {\em Phys. Lett.}, {\bf 42A}:313--314, (1972).

\bibitem{Uelt04}
D.~Ueltschi.
\newblock Cluster expansions and correlation functions.
\newblock {\em Mosc.\ Math.\ J.}, {\bf 4}:511--522, (2004).

\bibitem{Wu66}
T.T. Wu.
\newblock Theory of {Toeplitz} determinants and the spin correlations of the
  two-dimensional {Ising} model. {I}.
\newblock {\em Phys. Rev.}, {\bf 149}:380--401, (1966).

\end{thebibliography}
